\documentclass[sigconf,nonacm]{acmart}

\usepackage{amssymb}
\usepackage{pifont}
\usepackage{tikz}
\usepackage{multirow}
\usepackage{float}
\usepackage[most]{tcolorbox}
\usepackage{colortbl}
\usepackage{array}
\definecolor{trimblue}{HTML}{2A78D6}\definecolor{trimorange}{HTML}{EB6834}
\definecolor{trimaqua}{HTML}{1BAF7A}\definecolor{trimyellow}{HTML}{EDA100}
\definecolor{trimink}{HTML}{1F1F1E}\definecolor{trimmuted}{HTML}{6B6A64}\definecolor{trimpurple}{HTML}{8A5CC2}
\newcounter{findingno}
\newtcolorbox{finding}[1][]{enhanced,breakable,#1,sharp corners,frame hidden,colback=trimblue!5,
  borderline west={2pt}{0pt}{trimblue},left=6pt,right=4pt,top=2pt,bottom=2pt,
  before skip=4pt,after skip=5pt,fontupper=\small}
\newcommand{\FI}[1]{\stepcounter{findingno}{\color{trimblue}\textsc{Finding~\arabic{findingno}}}\quad\emph{#1}\par\smallskip}

\newcommand{\KN}{\par\noindent{\color{trimblue}$\rightarrow$}~}
\newtcolorbox{algbox}{enhanced,breakable,sharp corners,frame hidden,colback=white,
  borderline north={0.6pt}{0pt}{trimink},borderline south={0.6pt}{0pt}{trimink},
  left=2pt,right=2pt,top=3pt,bottom=3pt,before skip=5pt,after skip=5pt,fontupper=\small}
\newcommand{\ours}{\rowcolor{trimblue!7}}

\usetikzlibrary{arrows.meta,positioning,calc,backgrounds}

\colorlet{invblue}{trimblue}
\colorlet{feeorange}{trimorange}
\colorlet{repgreen}{trimaqua}
\colorlet{depred}{trimink}
\tikzset{
  chanbox/.style={draw=black!65,line width=0.5pt},
  feearrow/.style={->,draw=feeorange!80!black,line width=1.0pt},
  flowarrow/.style={->,draw=repgreen!70!black,line width=0.9pt},
  trimbox/.style={draw=black!70,rounded corners=2pt,line width=0.6pt},
}

\theoremstyle{acmdefinition}
\newtheorem{assumption}{Assumption}

\makeatletter\def\@authorfont{\large}\def\@affiliationfont{\large}\makeatother
\setcopyright{none}
\renewcommand\footnotetextcopyrightpermission[1]{}

\begin{document}

\title{\textsc{Trim}: Fee Level Matters More Than Skew in Lightning Routing}

\author{Huiping Sun$^{1}$\,\ding{41}, Peilu Guo$^{1}$, Yiming Zhu$^{2}$, Yueqi Wu$^{1}$, and Zhong Chen$^{2,3}$}
\affiliation{%
  \institution{$^{1}$School of Software \& Microelectronics, Peking University, Beijing, China\\
  sunhp@ss.pku.edu.cn, gpeilu@stu.pku.edu.cn, wuyueqi77@stu.pku.edu.cn\\
  $^{2}$School of Computer Science, Peking University, Beijing, China\\
  ymzhu25@stu.pku.edu.cn, zhongchen@pku.edu.cn\\
  $^{3}$School of AI and Liberal Art, Beijing Normal-Hong Kong Baptist University, Zhuhai, China}
  \country{}
}
\renewcommand{\shortauthors}{Sun et al.}

\begin{abstract}
Lightning routing nodes lend their channel balance for a posted fee. Operators'
tools move that fee with the balance around a hand-set level, and
which part of the fee drives revenue is unclear. We present \textsc{Trim},
which models a channel as a two-sided inventory with costly refills. Its optimal
fee is a \emph{level}, set by demand and replenishment, plus a balance \emph{skew},
the shadow price of liquidity. The level is first order: no fee rule beats the
ceiling demand and replenishment set, and on exactly solved channels the best flat fee comes within
$1.6\%$ of the optimum at the median, whereas a level off by a factor of two loses
tens of percent. The skew is second-order insurance: it adds at most what the
flow-balancing flat fee loses to refused forwards. So the level should be chosen: \textsc{Trim}'s rule traces the revenue--success frontier, posts one level on every channel, and adds a skew only where refills stay
frequent and costly. On a 2026 snapshot with synthetic payments,
no tuned tool beats that frontier by over a few percent, and a skew changes
gross revenue by a few percent.
\end{abstract}



\maketitle

\section{Introduction}
\label{sec:intro}

\begin{figure}[t]
\centering
\begin{tikzpicture}[font=\small,>={Stealth[length=4pt]},scale=0.92,
  show background rectangle,inner frame sep=5pt,
  background rectangle/.style={fill=trimblue!5,draw=trimblue!45,rounded corners=3pt,line width=0.5pt}]
  \draw[->,draw=black!55,line width=0.4pt] (0,0) -- (6.25,0);
  \node[below,font=\scriptsize] at (1.8,0){balance $b$};
  \draw[->,draw=black!55,line width=0.4pt] (0,0) -- (0,1.85) node[left,font=\scriptsize]{fee};
  \draw[draw=trimblue,line width=2.0pt] (0,0.55) -- (6,0.55);
  \draw[draw=feeorange!70,line width=0.7pt,dashed,domain=0.1:6,samples=40]
    plot (\x,{0.55+0.8*(exp(-0.85*\x)-exp(-0.85*3.6))});
  \draw[dotted,draw=black!45] (3.6,0) -- (3.6,0.55);
  \node[below,font=\scriptsize] at (0,0){$0$};
  \node[below,font=\scriptsize] at (3.6,0){$\bar b$};
  \node[below,font=\scriptsize] at (6,0){$c$};
  \node[font=\scriptsize\bfseries,text=trimblue!85!black,anchor=north east] at (6,0.5){level $L$};
  \node[font=\scriptsize,text=feeorange!80!black,anchor=west] at (0.3,1.35){skew};
  \begin{scope}[shift={(2.4,1.25)}]
    \draw[chanbox] (0,0) rectangle (3.6,0.32);
    \fill[invblue!35] (0,0) rectangle (2.1,0.32);
    \node[font=\scriptsize] at (1.05,0.16){balance $b$};
    \node[font=\scriptsize] at (2.85,0.16){$c-b$};
    \draw[->,draw=red!70!black,line width=0.7pt] (-0.55,0.16) -- (-0.04,0.16);
    \node[font=\scriptsize,text=red!70!black,anchor=south] at (-0.3,0.3){refill};
    \draw[feearrow] (1.65,-0.2) -- (0.95,-0.2);
    \node[font=\scriptsize,text=feeorange!75!black,anchor=east] at (0.9,-0.2){demand};
    \draw[flowarrow] (1.95,-0.2) -- (2.5,-0.2);
    \node[font=\scriptsize,text=repgreen!60!black,anchor=west] at (2.55,-0.2){replenishment};
  \end{scope}
\end{tikzpicture}
\caption{The \textsc{Trim} model. Inset: a channel as an inventory, drawn
down by demand and restored by replenishment or a costly refill. Plot: the
optimal fee, a level $L$ plus a skew.}
\label{fig:model}
\Description{A channel drawn as a bar split into the node's balance and the peer's side. Demand, priced by the node, depletes the balance; replenishment, priced by the peer, restores it; an on-chain refill restores it. Above, the optimal fee against the balance: a flat level plus a skew that rises as the channel drains.}
\end{figure}

\paragraph{Routing nodes and their fees.}
Payment channel networks (PCNs) such as the Lightning
Network~\cite{poon2016lightning} settle payments off-chain and route them
hop by hop through intermediary \emph{routing nodes}, which lend their
channel balance to strangers for a posted fee. About $5{,}500$
self-interested nodes and $27{,}000$ public channels
post these prices, $150$\,ppm at the median channel direction; as Lightning carries
small payments for online services and content~\cite{l402,nip57}, they set
what such payments cost and whether they succeed. Forwarding out of a channel draws down the node's side, which is replenished
only when payments flow back; a channel that runs dry stalls until reverse flow arrives or the operator pays for an
on-chain refill~\cite{khalil2017revive}.

\paragraph{How fees are set today.}
Operators who manage liquidity run
rule-based tools (\texttt{charge-lnd}~\cite{chargelnd}, feeadjuster~\cite{feeadjuster}, LNDg~\cite{lndg}) that raise a
channel's fee as its balance drains and lower it as the balance fills, around a
level set by hand or inherited from the deployed
fee. In 2019--2023 gossip the largest hubs reset that level $1.7$ times a year, by a median $2.3\times$, most channels at once (our measurement, App.~\ref{app:gossip}).
Which part of the fee drives a routing node's revenue?

\paragraph{\textsc{Trim}.}
A routing node is a two-sided inventory manager
(Fig.~\ref{fig:model}). Balance is inventory: demand depletes it, and
replenishment, priced by other nodes, restores it, both sides holding the
same asset~\cite{gallego1994optimal,avellaneda2008high}. Its optimal fee has two layers that do
different jobs. The \emph{level} sets how much demand the node prices in and whether
the channel drifts: at the flow-balancing fee demand matches replenishment and
the balance neither drains nor fills, while a level set too low drains the
channel and one set too high prices out demand. The \emph{skew}, the marginal value of balance, moves with the balance and only absorbs its fluctuations around that point. Neither layer can be read off public data: both depend
on how price-sensitive senders are and on flows a node does not observe. We call
this account \textsc{Trim}: a model of the channel, the two-layer fee it
implies, and the rule that follows.

\begin{finding}
{\color{trimblue}\textsc{Contributions}}\par\smallskip
\KN \emph{(1) A two-layer fee whose level is first order} (\S\ref{sec:single}--\S\ref{sec:hub}). We model a channel as a two-sided inventory with costly refills and derive its optimal fee, a level plus a skew. Under our assumptions no fee rule beats the revenue ceiling that demand and replenishment set, for any decreasing demand (Prop.~\ref{prop:ceiling}), and where demand outruns replenishment a level error costs in proportion to its size. At a hub the level stays uniform unless a channel drains (Prop.~\ref{prop:hub}).\par
\KN \emph{(2) The skew is second-order insurance} (\S\ref{sec:sg-value}). A skew adds at most what the flow-balancing flat fee loses to refused forwards (Prop.~\ref{prop:value}) and pays only where refills are costly.\par
\KN \emph{(3) A rule that makes the level a choice} (\S\ref{sec:trim}). Trace the revenue--success frontier with a few levels, pick a point and post it on every channel; add a skew only where refills stay frequent and costly.\par
\KN \emph{(4) Tuning mainly selects a level} (\S\ref{sec:eval}). On a 2026 snapshot with synthetic payments, uniform levels trace a revenue--success frontier, and our tuned reimplementations of the deployed tools lie on or near it: their balance rules add little at equal success.
\end{finding}

\paragraph{Results.}
On $144$ exactly solved channels a well-chosen flat fee is nearly optimal, within
$1.6\%$ at the median (at most $9.6\%$), while a level off by a factor of two loses tens of percent; on $2{,}592$ more with other size laws and demand curves the median gap is $1.1$--$2.3\%$.
On exactly solved hubs where a channel drains, per-channel levels, not a skew, recover half to three-quarters of what a uniform level loses.
In the network, no tuned tool beats the frontier of uniform levels by more than a few percent, levels per hub or per channel class give no gain, and a skew moves gross revenue by a few percent at most.
Tools win clearly only on 2019, with balance-aware senders, by pricing channels differently, not by a skew: deployed fees at their best uniform multiple do as well (\S\ref{sec:conclusion}). Code, raw outputs and a number checker will be made public upon acceptance.

\section{Background and Related Work}
\label{sec:background}

\begin{table}[t]
\caption{Related work. Alone: deployable by one node.}
\label{tab:related}
\Description{Table comparing families of related work and Trim by what moves the fee, how its level is set and whether one node can deploy it.}
\footnotesize
\setlength{\tabcolsep}{2.2pt}\renewcommand{\arraystretch}{0.95}
\begin{tabular}{@{}llll@{}}
\toprule
Family (work) & Fee moves with & Level & Alone \\
\midrule
Deployed tools {\scriptsize\cite{chargelnd,feeadjuster,lndg,clboss}} & balance rules & by hand & \checkmark \\
Pricing imbalance {\scriptsize\cite{distasi2018routing,vanengelshoven2021merchant,wang2024fence,sankagiri2025pricing,desilva2023mpcfc,asgari2022dyfen}} & imbalance, duals & absent/implicit & mostly \\
Pricing inventory {\scriptsize\cite{gallego1994optimal,talluri2004theory,chen2004infinite,ata2005drift,avellaneda2008high,gueant2013dealing}} & marginal value & markup & --- \\
Other {\scriptsize\cite{khalil2017revive,mishra2026feeder,ge2022shaduf,wu2026sluice,avarikioti2024musketeer,bastankhah2023r2,wang2022riding,avarikioti2020ride,ersoy2020profit,carotti2024rational}} & fixed/strategies & --- & --- \\
\rowcolor{trimblue!10}\textsc{Trim} (this paper) & shadow price & explicit & \checkmark \\
\bottomrule
\end{tabular}
\end{table}

\paragraph{Lightning fees.}
Each channel direction carries a fee set by the node on whose side a payment leaves, a base charge plus a proportional (ppm) rate~\cite{poon2016lightning}; since LND 0.18 a node can also post an inbound fee, at first only a discount~\cite{lndinbound}. Senders pick fee-minimizing routes~\cite{pickhardt2021flows}, which LND weighs against the success probabilities its \emph{mission control} learns from failures, so a node's demand is price-sensitive.

\paragraph{Deployed tools.}
Implementations ship static defaults (LND: $1{,}000$\,msat $+\,1$\,ppm~\cite{lnddefaults}). \texttt{charge-lnd} interpolates linearly in the local share~\cite{chargelnd}, feeadjuster scales a base rate by $50^{0.5-x}$~\cite{feeadjuster}, LNDg steps fees by flow and liquidity thresholds~\cite{lndg}, and CLBOSS scales competitors' median fee by size, balance and price~\cite{clboss}. LND's developers report that success falls with a channel's fee, yet only $10\%$ of nodes priced by liquidity in 2024~\cite{lnd2024pathfinding}. All tune the skew and inherit or search the level.

\paragraph{Pricing imbalance.}
Proposed mechanisms surcharge forwards that worsen balance~\cite{distasi2018routing,vanengelshoven2021merchant}, post the dual prices of an online packing program~\cite{wang2024fence}, or, in \textsc{Debt}~\cite{sankagiri2025pricing}, run network-wide dual descent toward the welfare-maximizing flow; node-profit controllers~\cite{desilva2023mpcfc,asgari2022dyfen} are not deployed. \textsc{Debt}'s prices are multipliers of flow conservation, our $\hat s$ one node's marginal value of its own balance, which it can post alone; the incentive component \textsc{Debt} leaves to practice is the level. Throughout, the balance term is the lever and the level is given or implicit.

\paragraph{Pricing an inventory.}
An inventory's optimal price is the marginal value of a unit plus the markup $1/\kappa$~\cite{gallego1994optimal,talluri2004theory}, with $(s,S,p)$ policies under fixed ordering costs~\cite{chen2004infinite}, drift control~\cite{ata2005drift,yao2017joint} and inventory-skewed quotes for market makers~\cite{avellaneda2008high,gueant2013dealing}; static prices are near-optimal in closed networks~\cite{banerjee2022pricing}. A channel's inventory is bounded on both sides and replenished at prices set by others, and flow balance is its zero-drift point.

\paragraph{Other related work.}
Rebalancing moves liquidity after imbalance forms~\cite{khalil2017revive,mishra2026feeder,ge2022shaduf,wu2026sluice,avarikioti2024musketeer} and admission control protects it~\cite{bastankhah2023r2}, both with fees fixed; hub equilibria and channel games~\cite{wang2022riding,avarikioti2020ride}, routing profitability~\cite{carotti2024rational,ersoy2020profit} and depletion levers~\cite{pickhardt2026theory} concern network outcomes. We study one node's split of its own fee; where levels settle when all reprice we leave open (App.~\ref{app:level}).

\paragraph{Positioning.}
\textsc{Trim} is a two-layer account of one node's fee and the rule that follows, not a routing, rebalancing or equilibrium design. No prior work we know derives the two layers and asks which one drives revenue (Table~\ref{tab:related}). Relative to inventory pricing, what we add is a revenue ceiling that binds every fee rule when replenishment is priced by others, a bound on what a skew adds over the best flat fee, and a diagnosis of the tools operators run.

\section{\textsc{Trim}: A Channel as an Inventory}
\label{sec:model}

\paragraph{Scope and state.}\label{sec:model-informal}
We model one channel of a routing node and its proportional rate, taking the
node's other channels as part of the network and its base fee as deployed
(Fig.~\ref{fig:model}); it omits time-locks, HTLC limits, multi-part payments and channel closures (Simplifications below). The state is the node's balance $b\in[0,c]$ on its side of a channel of capacity
$c$, its \emph{inventory}. \emph{Demand}, the forwards that leave through that
side, arrives at rate $\lambda(f)$ that falls with the proportional fee $f$, the
control; \emph{replenishment}, the forwards that enter, arrives at rate $\mu$.
A node prices only payments that leave through its side; those that enter are
priced by the peer and upstream nodes, so per channel $\mu$ is exogenous. The
canonical demand is $\lambda(f)=\Lambda e^{-\kappa f}$, the logit form of
least-cost route choice, with markup $1/\kappa$. Offered demand in our simulator
is close to it near deployed fees ($R^2=0.95$, App.~\ref{app:evalx}). Far from them it mixes
captive and competed demand, whose markup can grow faster than the fee (for our fitted mixture, between $52$ and $1{,}052$\,ppm), so we do not assume log-concavity.

\paragraph{Costs and objective.}
Unlike a market maker's, the inventory carries no price risk: both sides hold the
same asset. Risk enters only at the boundaries, where a forward larger than
the balance is refused at a goodwill cost $c_{\mathrm r}$ and an on-chain refill
of size $\Delta$ costs $K+k\Delta$. The node maximizes its long-run average net
revenue: fees on served demand, less refusal penalties and refill costs.

\paragraph{Formalization.}\label{sec:model-formal}
Demand and replenishment are marked point processes $N^{\mp}$ with
size laws $G$ and $G^{+}$, and refills add $\Delta_i>0$ at times $\tau_i$, so
\begin{equation}
\begin{aligned}
  b_t = b_0
  &-\!\int_0^{t}\!\!\int\!\xi\,\mathbf 1\{\xi\le b_{u^-}\}\,N^{-}(du,d\xi)\\
  &+\!\int_0^{t}\!\!\int\!\eta\,\mathbf 1\{\eta\le c-b_{u^-}\}\,N^{+}(du,d\eta)
  +\!\sum_{\tau_i\le t}\!\Delta_i .
\end{aligned}
\label{eq:inv-dyn}
\end{equation}

\begin{assumption}[Demand]
\label{ass:demand}
$\lambda:[0,\infty)\to(0,\infty)$ is $C^1$ with $\lambda'<0$ and
$f\lambda(f)\to0$ as $f\to\infty$; its markup is $\psi(f):=-\lambda(f)/\lambda'(f)$.
\end{assumption}

\begin{assumption}[Depletion cost]
\label{ass:depletion}
(i) A refused forward of demand costs a goodwill penalty $c_{\mathrm r}\ge0$
plus the forgone fee; refused replenishment costs nothing directly.
(ii) A refill of size $\Delta$ costs $K+k\Delta$, with $K>0$ and $k\ge0$.
(iii) Service of feasible forwards is compulsory: the node prices demand away rather than refusing a feasible forward, so refusals occur only when a forward exceeds the balance.
\end{assumption}

A policy $\pi$ is a predictable fee process together with a refill rule; a
Markov policy sets $f_t=\varphi(b_{t^-})$. Let $N^{\mathrm r}_t$ count refused
forwards of demand. The expected discounted net revenue is
\begin{equation}
  J_\rho(\pi)=\mathbb E^{\pi}_{b_0}\Big[\!\int_0^{\infty}\!\!e^{-\rho t}
  \big(f_t\,\xi\,dN^{-}_{\mathrm{served}}-c_{\mathrm r}\,dN^{\mathrm r}_t\big)
  -\sum_i e^{-\rho\tau_i}(K+k\Delta_i)\Big],
  \label{eq:objective}
\end{equation}
with $\rho>0$ and value function $V(b)=\sup_\pi J_\rho(\pi)$. For the
average net revenue $g=\sup_\pi\liminf_{T\to\infty}\frac1T\,\mathbb E^\pi[\cdot]$
the optimality equation has the same form, with $\rho V$ replaced by $g$ and $V$
by the bias $h$~\cite{hernandezlerma1996markov}. We prove the structural results below for $\rho>0$ and take the average criterion as
their limit $\rho\to0$, which we assume exists (a numerical assumption: relative value iteration
converges on every instance we solve).

\paragraph{Simplifications.}
The model has no time-locks, HTLC limits, multi-part payments or channel
closures (the simulator has HTLC and MPP limits as options). A refill restores the balance at once, and demand depends on the channel's own fee only. We do not model inbound fees, with which a node would also price its replenishment; in an exact extension the inbound fee is a second level, and level errors still cost more than state dependence adds (App.~\ref{app:single}).

\section{The Channel Fee: A Level plus a Skew}
\label{sec:single}

For one channel, the level is $L=f^{*}(b^{*})$, the optimal fee at the balance $b^{*}$
where the drift vanishes ($c$ if none), and the skew is the rest, $f^{*}(b)-L$; a flat fee
is a level alone. This split is a definition; what we prove is what a wrong level
costs and what any balance-dependent fee can add over the best flat one.

We derive the two layers (\S\ref{sec:sg-foc}), then the level (\S\ref{sec:sg-level}) and the skew (\S\ref{sec:sg-value}), check both on exactly solved instances and map where the skew pays (\S\ref{sec:sg-exact}). Each proposition's title says whether it is exact or uses the diffusion approximation, which we check against exact solutions. Proofs are in
App.~\ref{app:single}; Table~\ref{tab:notation} lists the notation. Write $\bar G(x)=\mathbb P(\xi>x)$,
$m(x)=\mathbb E[\xi\,\mathbf 1\{\xi\le x\}]$, $m_1=\mathbb E\xi$,
$m_2=\mathbb E\xi^2$ and $m_1^{+}=\mathbb E\eta$.

\begin{figure*}[t]
\centering
\includegraphics[width=\textwidth]{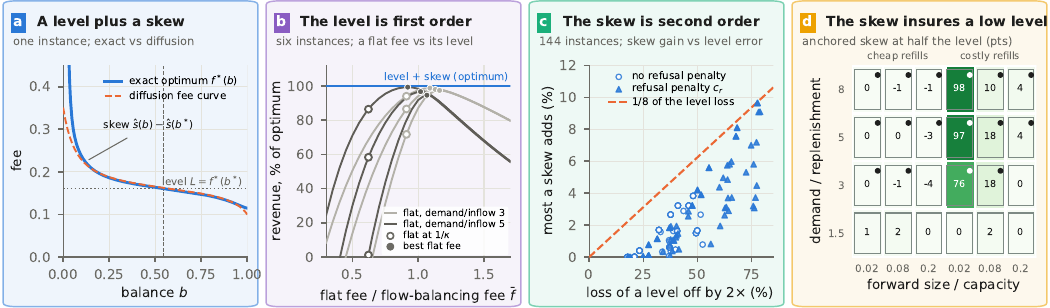}
\caption{The exact single channel: the level is first order, the skew second-order insurance. (a) One instance: the optimal fee, a level $L=f^{*}(b^{*})$ plus a skew, against the diffusion curve. (b) Six instances: a flat fee's revenue against its level, relative to the optimum (filled: best flat fee; hollow: the markup $1/\kappa$). (c) $144$ instances: the most any balance-dependent fee adds over the best flat fee, against the loss of a level off by a factor of two. (d) What the anchored skew~\eqref{eq:trim} recovers at half the best level, by ratio of demand to replenishment, forward size and refill cost (cell medians; dot: a cell where the half level drains some channel into a net loss).}
\label{fig:theory}
\Description{Four panels. (a) Fee against balance for one exactly solved channel: a nearly flat fee that rises steeply near an empty channel, with the diffusion curve on top of it. (b) Revenue of a flat fee against its level for six instances, peaking within a few percent of the optimum near the flow-balancing fee and falling steeply away from it. (c) Scatter of 144 instances: what a skew adds against the loss from a level off by a factor of two; all points lie below a line at one eighth. (d) Heat map of what the anchored skew recovers at half the best level: most of the optimum only where refills are costly and forwards small, little or nothing elsewhere.}
\end{figure*}

\begin{table}[t]
\centering
\caption{Notation.}
\label{tab:notation}
\Description{Table of the model's symbols.}
\footnotesize
\setlength{\tabcolsep}{2.5pt}\renewcommand{\arraystretch}{0.95}
\begin{tabular}{@{}l>{\raggedright\arraybackslash}p{0.17\textwidth}l>{\raggedright\arraybackslash}p{0.17\textwidth}@{}}
\toprule
$b,\ c$ & balance, capacity & $\lambda(f)$ & demand rate, $\Lambda e^{-\kappa f}$ \\
$\mu$ & replenishment rate & $\psi$ & markup, $1/\kappa$ \\
$\xi,\ \eta$ & forward sizes, out / in & $m_1,m_2;\ m_1^{+}$ & moments of $\xi$; mean of $\eta$ \\
$c_{\mathrm r}$ & refusal penalty & $K+k\Delta$ & refill cost \\
$L,\ \hat s(b)$ & level, skew & $\bar f$ & flow-balancing fee \\
$g;\ R^{*}$ & average revenue; its ceiling at flow balance & $\nu_0$ & variance rate of $b$ \\
$b^{*};\ \hat s^{*}$ & flow-balance point; its shadow price & $\Phi(s),F(s)$ & fee income at shadow price $s$; plus inventory value \\
\bottomrule
\end{tabular}
\end{table}

\subsection{Two layers}
\label{sec:sg-hjb}\label{sec:sg-foc}

For a fee $f$ and a value function $w$ on $[0,c]$, the controlled generator is
\begin{equation}
\begin{aligned}
  (\mathcal L^{f} w)(b)=
  &\;\lambda(f)\Big[\int_{0}^{b}\!\big(f\xi+w(b-\xi)-w(b)\big)G(d\xi)
      -c_{\mathrm r}\bar G(b)\Big]\\
  &+\mu\int_{0}^{c-b}\!\big(w(b+\eta)-w(b)\big)G^{+}(d\eta),
\end{aligned}
\label{eq:gen}
\end{equation}
and the refill operator is
$(\mathcal M w)(b)=\sup_{0<\Delta\le c-b}\{w(b+\Delta)-K-k\Delta\}$. The first
bracket is a served forward's fee plus its change in inventory value, less
the goodwill cost of refusals; the second term is replenishment.

\begin{proposition}[Optimality eq\kern0ptuation; exact]
\label{prop:single-opt}
Under Assumptions~\ref{ass:demand}--\ref{ass:depletion}, $V$ is the unique
bounded solution of the quasi-variational inequality
\begin{equation}
  \min\big\{\rho V(b)-\sup_{f\ge0}(\mathcal L^{f}V)(b),\;
            V(b)-(\mathcal MV)(b)\big\}=0,\quad b\in[0,c].
  \label{eq:qvi}
\end{equation}
\end{proposition}

Only the demand bracket of \eqref{eq:gen} depends on $f$. Define the
\emph{shadow price of liquidity}
\begin{equation}
  \hat s(b)=\frac{\int_0^{b}\big(V(b)-V(b-\xi)\big)G(d\xi)
        +c_{\mathrm r}\bar G(b)}{m(b)},
  \label{eq:shadow}
\end{equation}
the inventory value the node gives up per unit forwarded, including the
expected refusal cost.

\begin{proposition}[Optimal fee structure; exact]
\label{prop:fee-struct}
In the continuation region the optimal fees maximize $\lambda(f)\,(f-\hat s(b))$; each positive one solves
\begin{equation}
  f^{*}(b)=\hat s(b)+\psi\big(f^{*}(b)\big),
  \label{eq:fee-foc}
\end{equation}
and the largest, $f^{*}(b)$, is nondecreasing in $\hat s(b)$; it is the only root if $f-\psi(f)$
increases (log-concave $\lambda$), else it can jump. For $\lambda=\Lambda e^{-\kappa f}$,
$f^{*}(b)=\hat s(b)+1/\kappa$.
\end{proposition}

This is the classical structure of inventory
pricing~\cite{gallego1994optimal}, congestion-dependent pricing of
queues~\cite{low1974optimal,paschalidis2000congestion,maglaras2003pricing} and
risk-neutral market making~\cite{gueant2013dealing}; what is specific to a channel is how the
shadow price depends on the balance.

\subsection{The level is first order}
\label{sec:sg-level}

By first order we mean that a level off by a factor costs more than the skew can add (\S\ref{sec:sg-value}): in proportion to the error where flow balance binds, and $1-re^{1-r}$ ($26\%$ at $r=2$) at the markup peak. Let $\bar f$ be the \emph{flow-balancing fee}, at which demand
balances replenishment in volume, $m_1\lambda(\bar f)=\mu m_1^{+}$; under
exponential demand $\bar f=\frac1\kappa\ln\frac{\Lambda m_1}{\mu m_1^{+}}$. When flow
balance lies above the refill trigger, in the diffusion approximation the level is $\max\{1/\kappa,\bar f\}$:
the shadow price there is $\max\{0,\bar f-1/\kappa\}$, so the level is the markup $1/\kappa$ when replenishment keeps
up with demand at that fee, else $\bar f$. Either way it is set by
demand and replenishment alone and scales as $1/\kappa$.

\paragraph{A revenue ceiling at flow balance.}
Write $\Phi(s):=m_1\sup_{f\ge0}\lambda(f)(f-s)$ for the fee income of a node that values liquidity at $s$, and $F(s):=\Phi(s)+\mu m_1^{+}s$ for its fees plus inventory value. $F$ is convex for any demand, $\Phi$ being a supremum of affine functions of $s$
(for exponential demand $\Phi(s)=\frac{m_1\Lambda}{\kappa e}e^{-\kappa s}$);
where some $\hat s^{*}\in[0,k]$ has $\bar f$ as its best fee, its minimum over $[0,k]$ is attained there, where demand balances replenishment
($\hat s^{*}=\bar f-1/\kappa$ for exponential demand), and equals
$R^{*}:=\mu m_1^{+}\bar f$, the fee income of a channel held exactly at flow
balance.

\begin{proposition}[Revenue ceiling; exact]
\label{prop:ceiling}
Let the long-run rate of replenishment volume served be at most $\mu m_1^{+}$,
whether $\mu$ is constant, state-dependent or set by the network. Under
Assumptions~\ref{ass:demand}--\ref{ass:depletion}, every policy, with
any fee rule and any refills, earns $g\le F(s)$ for every $s\in[0,k]$, so
$g\le F_k:=\min_{0\le s\le k}F(s)$. For any decreasing demand under which $\bar f$ maximizes $\lambda(f)(f-s)$ for some $s=\hat s^{*}\in[0,k]$
(for exponential demand, $\Lambda m_1\ge e\,\mu m_1^{+}$ and $k\ge\hat s^{*}$), $F_k=R^{*}$: no policy earns
more than the fee income of a channel held at flow balance.
\end{proposition}

In words: no fee rule, flat or balance-dependent, beats the ceiling that demand and replenishment set.

\paragraph{A wrong level.} It drains the channel or starves it. Let $d(f)=\mu
m_1^{+}-m_1\lambda(f)$ be the drift of the balance under a constant fee $f$, so
$d(\bar f)=0$. Below $\bar f$ a static fee drains the channel at rate $|d(f)|$, above it
starves it. A wrong level is also costly. When demand outruns replenishment, the markup is the
wrong level: on six exactly solved instances a flat fee there loses
$6$--$99\%$ (Fig.~\ref{fig:theory}(b)).

\subsection{The skew is second-order insurance}
\label{sec:sg-closed}\label{sec:sg-value}

\paragraph{Why its revenue value is small.}
Let $\chi$ be the scale over which the marginal value changes and
$\varepsilon:=m_2/(m_1\chi)$. When $\varepsilon\ll1$ and the balance lies well above $0$,
$\hat s=h'+O(\varepsilon)$ and, in the diffusion approximation~\cite{ata2005drift,cao2018drift}
(checked against the exact model in App.~\ref{app:single}), the average-reward optimality equation reduces to
\begin{equation}
  g=\Phi(h')+\mu m_1^{+}h'+\tfrac{\nu_0}{2}h'',
  \label{eq:hjb-diff}
\end{equation}
where $\nu_0=\lambda(\bar f)m_2+\mu\,\mathbb E\eta^2$ is the variance rate of
the balance at flow balance, held constant. Refused replenishment at a full channel is costless, so $h'(c)=0$.

In the diffusion limit the shadow price falls with the balance (Proposition~\ref{prop:skew}).

\paragraph{What it can add.} In words: over the best flat fee, a skew adds at most what that fee loses to refused forwards (part~(i)); its job is to keep a wrong level from draining the channel (part~(ii)).

\begin{proposition}[Value of state dependence; partly diffusion]
\label{prop:value}
Under exponential demand with
$\Lambda m_1>e\,\mu m_1^{+}$:
\emph{(i)} let $k\ge\hat s^{*}$; then the optimal policy earns $g<R^{*}$. If demand and
replenishment forwards share one size law, the static fee $\bar f$ without refills
earns exactly $R^{*}-\Gamma$, with $\Gamma:=\lambda(\bar f)\big(\bar f\,\mathbb E[\xi(\xi\wedge c)]+c_{\mathrm r}\mathbb E[\xi\wedge c]\big)/c$,
which is $(\bar f+c_{\mathrm r}m_1/m_2)\,\nu_0/(2c)$ when $\xi\le c$, as it is for any
size laws in the diffusion regime; the optimal policy gains at most $\Gamma$ over the best
static fee, and in the diffusion regime, if $b^{*}$ exists, the skew slope $a=2(R^{*}-g)/\nu_0$ at flow balance satisfies $a\,c\le\bar f+c_{\mathrm r}m_1/m_2$.
\emph{(ii)} A static fee $f_0$ loses (or must refill) net volume at rate at least $|d(f_0)|$, by flow conservation, and at rate $|d(f_0)|\,(1+o(1))$ in the diffusion regime with $|d(f_0)|\,c\gg\nu_0$. An affine skew with the same level,
$f(b)=f_0-a(b-b_0)$, has zero drift at
$b_{\mathrm{eq}}=b_0+(f_0-\bar f)/a$; if $b_{\mathrm{eq}}\in(0,c)$ its loss
rate is exponentially small in
$\theta\min(b_{\mathrm{eq}},c-b_{\mathrm{eq}})^2/\nu_0$, with
$\theta=\kappa\mu m_1^{+}a$ its rate of mean reversion (App.~\ref{app:single}).
\end{proposition}

Against the \emph{best} static fee, state dependence adds little when payments
are small relative to capacity ($\nu_0/c$) and refusals cost little
($c_{\mathrm r}m_1/m_2$), as in revenue
management~\cite{gallego1994optimal,besbes2022static,banerjee2022pricing}; on six exactly solved instances the bound is $4\%$ of $R^{*}$ without $c_{\mathrm r}$ and $7$--$22\%$ with it, while the best constant fee loses $0.7$--$5.4\%$. The slack is of higher order in $\delta=\nu_0/(2c\mu m_1^{+})$: the best flat fee loses $\delta$ times a bracket, the optimum $\delta^{2}$ times one (App.~\ref{app:single}). What a skew does is part~(ii): a static
fee that misses $\bar f$ drains or starves the channel, while a skew turns the
level error into a shift of the operating balance, seen as fewer refills.

\paragraph{The anchored skew.} With the level held at $L$, a node can post
\begin{equation}
  f(b)=L+\hat s_L(b)-\hat s_L(\bar b),
  \label{eq:trim}
\end{equation}
where $\hat s_L$ is the shadow price under the flat fee $L$, computable from the channel's own flows at $L$, and $\bar b$ the channel's time-average balance. The anchor keeps
the fee at $L$ in the typical state: the skew cannot move the level. \S\ref{sec:sg-exact} checks it on the exact instances.

\subsection{Exact check}
\label{sec:sg-exact}

We solve the jump model exactly by relative value iteration ($400$ balance steps) on $144$ instances that vary price sensitivity, demand over replenishment, forward size, refill cost and refusal penalty (App.~\ref{app:single}). The bound of Proposition~\ref{prop:value}(i) applies to the $72$ that meet $0<\hat s^{*}\le k$; the computations below cover all $144$. The optimal fee is a level plus a skew (Fig.~\ref{fig:theory}(a); $\kappa=10$, demand five times replenishment, $K=0.05$), and the level carries the revenue (Fig.~\ref{fig:theory}(b)). Over the $144$ instances the best flat fee comes within $1.6\%$ of the optimum at the median and $9.6\%$ at most. A flat fee at half or twice the best level loses at least $17.6\%$ ($40.5\%$ at the median), at least $8$ times what any balance-dependent fee adds on the same channel (Fig.~\ref{fig:theory}(c); worst ratio $1/8.05$, median $1/25$). Below $e$ times replenishment the peak is smooth: revenue at $r$ times the peak level is $re^{1-r}$ of it. Above, the level sits at the kink $\bar f$ and errors cost first order in fee income, about $1-r$ below it, where the channel drains, and $1-re^{-\kappa\bar f(r-1)}$ above; on the grid a flat fee $20\%$ below or $25\%$ above the best one loses $9.5$ and $8.9\%$ of the optimum at the median, and $17.0$ and $21.5\%$ with demand eight times replenishment. Beyond exponential sizes and demand, on $2{,}592$ exact instances with lognormal or Pareto forward sizes and mixed or linear demand (App.~\ref{app:single}), the best flat fee stays within $1.1$--$2.3\%$ of the optimum at the median, and a level off by a factor of two loses $20.6$--$61.1\%$; only under the mixed demand does it lose less than the flat fee's gap, on $144$ of $864$ instances.

The anchored skew~\eqref{eq:trim} does not repair a wrong level. At the best flat fee it cuts the gap from $1.6$ to $0.5\%$ at the median, but at twice that fee it changes nothing at the median and at half it gains $0.3$ points. It pays only on the $28$ instances where the half level drains the channel into a net loss, cutting that net loss from $154$ to $80\%$ of the optimum at the median. It is insurance against a level set too low where refills are costly. Scaling its shadow price by $\beta=0.5$, $0.8$ or $1.25$ changes the loss at the best level by at most $1.1$ points from $\beta=1$ on ten instances (App.~\ref{app:single}). The exact solutions thus support both roles: the best flat fee is near the optimum, and a factor-two level error costs far more than any skew adds.

\paragraph{When does the skew pay?}\label{sec:regimes}
$\kappa$ sets the level's scale and leaves the ordering; forward size sets the skew's bound, and the ratio of demand to replenishment sets $\bar f$. What a
skew adds grows with the ratio of demand to replenishment, from $0.0$ to $3.2\%$
at the median, and with forward size, from $0.8$ to $2.7\%$; with a refusal
penalty it reaches $9.6\%$, without one $3.8\%$. The skew's insurance has a region of its own: at half the
best level it recovers most of the optimum only where refills are costly and
forwards small against capacity (Fig.~\ref{fig:theory}(d)); elsewhere it
recovers at most $18$ points or loses a few.

\section{The Level at a Hub}
\label{sec:hub}\label{sec:sg-hub}\label{sec:hub-prop}

In a hub the channels feed each other, so a channel's replenishment is no longer
exogenous; this section couples them.

\begin{figure}[t]
\centering
\includegraphics[width=\columnwidth]{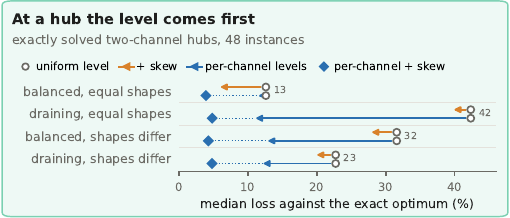}
\caption{At a hub the level comes first. Exactly solved two-channel hubs: the median loss against the exact optimum of the best uniform level, and where a skew on it (orange) or per-channel levels (blue; diamond: with each channel's own skew) take it, by flow balance and demand shape.}
\label{fig:hub}
\Description{For four kinds of two-channel hub, arrows from the uniform level's median loss: a skew on the uniform level moves it little where a channel drains or shapes differ, per-channel levels move it to 11 to 13 percent, and per-channel levels with skews to about 4 to 5 percent.}
\end{figure}

A forward enters a hub through one channel and leaves, and is priced, through
another: a channel's replenishment $\mu_i$ is exogenous to its own fee but
falls with the hub's other fees. Let channel $j$ face demand
$\lambda_j(f)=\Lambda_j\,\ell(f)$ with a common shape $\ell$, and let a share
$\pi_{ij}$ of the forwards leaving through $j$ enter through $i$, so
$\mu_i=\sum_{j\ne i}\pi_{ij}\lambda_j(f_j)$. We replace each channel's replenishment by its time
average under the posted fees, a mean-field decoupling that the exactly solved hubs below test.

\begin{proposition}[Hub level; mean-field, (ii) exact, (v) conjectured under coupling]
\label{prop:hub}
Let the hub post one level $L$ on every channel and a skew on each, with $\mu_i$ decoupled as above. Then
\emph{(i)} its average revenue is $\sum_i g_i(L;\mu_i(L))$ and, given $L$,
channel $i$'s best skew is that of Proposition~\ref{prop:fee-struct} at
$\mu_i(L)$;
\emph{(ii)} under flat fees at $L$, $\mu_i(L)/\lambda_i(L)$ does not depend on $L$: the level scales a
channel's inflow and outflow alike and leaves the sign of its drift, relative to
its flow, unchanged;
\emph{(iii)} where refills and refusals are rare at $L$,
$\sum_i g_i=L\,m_1\Lambda\,\ell(L)+o(1)$ with $\Lambda=\sum_i\Lambda_i$, so
the best level $L^{*}=\arg\max_L L\ell(L)$ is uniform and is each channel's own
best level, whatever the $\Lambda_i$, $\pi$ and $\mu_i$; for
$\ell(f)=e^{-\kappa f}$, $L^{*}=1/\kappa$ and a level $rL^{*}$ earns the
fraction $r\,e^{1-r}$ of the peak, $0.74$ at $r=2$ and $0.82$ at $r=1/2$;
\emph{(iv)} levels differ across channels only where refills bind, through $\bar f_i$ at $\mu_i(f_{-i})$ (channel $i$ drains when $\bar f_i>L^{*}$) (a fixed point exists, not necessarily unique, with fees capped where demand is negligible), or where the shapes differ;
\emph{(v)} the skews gain at most
$\sum_i(\bar f_i+c_{\mathrm r}m_1/m_2)\,\nu_{0,i}/(2c_i)$
(Proposition~\ref{prop:value}(i) per channel, under its conditions; with coupled
replenishment a conjecture, which we test).
\end{proposition}

Part~(iv) gives the two cases for per-channel levels, binding refills and price
discrimination across channels; neither is a skew.

\paragraph{Exactly solved hubs.}
Where a channel drains or demand shapes differ, per-channel levels, not a skew
on a uniform level, recover what a uniform level loses (Fig.~\ref{fig:hub}).

\emph{Setup.} We
check this on hubs small enough to solve exactly: two or three channels and the
rest of the hub, a deep channel, exchange forwards, each moving its size from
the exit's balance to the entry's, so the state is the joint balance. Part of each channel's traffic goes to the rest of the hub, so total balance is not conserved. Channel~$1$ drains and
channel~$2$ fills, in a ratio no uniform level changes. Against the optimum (policy iteration) we evaluate the best uniform level,
per-channel levels and each with the anchored skew (App.~\ref{app:single},
Table~\ref{tab:hub}). These instances test the mechanism, not real hubs.

\emph{Balanced flows, equal shapes.} A uniform level is right: the channels' own best levels lie within $3\%$ of the uniform level and add at most $2.0$ points. The uniformity of (iii) thus holds although refills are not rare, but its $o(1)$ does not: state dependence is worth $11$--$16\%$ here, as each forward moves two balances and is $5$--$10\%$ of capacity.

\begin{figure*}[t]
\centering
\includegraphics[width=\textwidth]{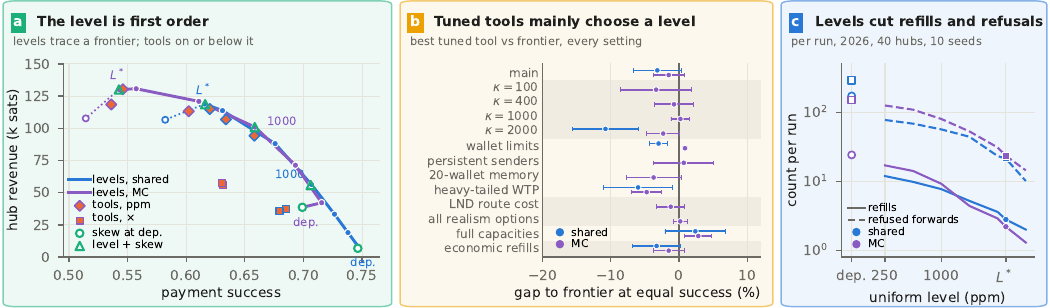}
\caption{Network evaluation on the 2026 snapshot, $40$ largest hubs: the level is first order, and tuned tools mainly choose a level. (a) Hub revenue against payment success: uniform levels (lines; hollow: above $L^{*}$), tuned tools (diamonds: ppm ranges; squares: multiples) and a skew at deployed fees or on a level (circles, triangles). (b) The best tuned tool against the frontier at its own success, in each of the $19$ settings and routing models ($95\%$ intervals; WTP: willingness to pay). (c) Refills (solid) and refused forwards (dashed) per run against the uniform level; hollow: deployed fees; filled: $L^{*}$.}
\label{fig:network}
\Description{Three panels. (a) Hub revenue against network payment success under two routing models: uniform levels trace a frontier, tuned tools sit on or below it, and level-plus-skew points coincide with the level points. (b) For each 2026 setting, the best tuned tool's revenue relative to the frontier, with intervals, mostly at or below zero. (c) Refills and refused forwards per run on log scales: deployed fees have the most, and both counts fall as the uniform level rises.}
\end{figure*}

\emph{A draining channel or different shapes.} The uniform level loses $22$--$42\%$ of the optimum at the median and per-channel levels cut that to $11$--$13\%$. A skew on the uniform level does not recover it (at worst the loss reaches $163\%$ of the optimum, a net loss), as (iv) says. On per-channel levels the skew adds a further $7$ points, more than on one channel, as the two balances move together. A factor-two error in the uniform level costs a further $18\%$ at the median (at least $10\%$), $1.5$ to $2.6$ times (minimum, median) what that skew adds.

\emph{Decoupling and the bound.} Per-channel levels with each channel's own skew come within $4.7\%$ of the optimum at the median, the cost of the decoupling in (i). Part~(v) applies to the $24$ two-channel and $4$ three-channel hubs with equal shapes: the optimum's gain over per-channel levels reaches $0.41$--$0.83$ of its bound on the $8$ two-channel ones that meet $k\ge\hat s^{*}$, exceeds it on $2$ of the $16$ that do not ($1.18$ and $1.51$), and all four three-channel ones hold. Three-channel hubs repeat these cases (App.~\ref{app:hubx}). At a hub too the level, set per channel where one drains, is the first-order choice; \S\ref{sec:trim} turns this into a rule.

\section{The \textsc{Trim} Rule: Choosing the Level}
\label{sec:trim}
\label{sec:sg-trim}

The level is first order and the skew second-order insurance (\S\ref{sec:single}--\S\ref{sec:hub}), so the rule makes the level an explicit, measured choice and the skew an exception; \S\ref{sec:eval} tests the claims behind it. The level cannot be computed from public data alone, as $\kappa$ and the senders' demand are not observed before a scan (given them, a model predicts it, \S\ref{sec:ev-peak}); it also sets how many payments the node carries.

\paragraph{Choosing the level.}
The operator first stops any other tool that writes the ppm field (\texttt{charge-lnd}, feeadjuster, LNDg's auto-fees, CLBOSS). It then posts one uniform level $L$, which moves the node along a revenue--success frontier (\S\ref{sec:eval}). Levels above the revenue peak $L^{*}$ earn no more at lower success, so they are dominated; senders' wallet fee limits cap the level where they bind (App.~\ref{app:x26}). The operator traces its own frontier by posting a few uniform levels, skew off, and picks a point by revenue, success and welfare. One choice is the highest level that keeps network payment success within a chosen margin: on our frontier $1{,}000$\,ppm costs $4$ points of network success. A \emph{large node}, one that can hold each level for tens of thousands of its own forwards, traces its frontier: comparing revenue, each level must be held long enough to tell it from the next, about $32$--$40$ thousand own forwards per level for levels $4\%$ apart~\cite{besbes2009dynamic}. Counting forwards is cheaper: probe $L/\sqrt2$ and $L\sqrt2$, estimate the local elasticity $\hat\kappa$ of the forward count and move to $1/\hat\kappa$, the first-order condition $L=1/\kappa(L)$, which needs a few hundred pooled forwards per probe. In our network it finds the peak when the hubs pool their forwards, but a lone hub's probes hold too few and it overshoots (\S\ref{sec:ev-choose}). A \emph{small node}, or any node searching alone, takes the level a model of its captive demand predicts (\S\ref{sec:ev-peak}), or a similar hub's, and keeps the skew off.

\paragraph{Adding the skew.} By default the rule posts no skew: with refills priced per amount it has no measurable net value in our network. Only where refills are frequent and costly at a level already chosen does it post the anchored skew~\eqref{eq:trim}.
Net of the fixed-trigger refills, the skew at deployed fees gains $34.2\pm6.3$k (shared memory) and $1.2\pm0.6$k (mission control) sat per run at $K=1{,}000$ (skew tables fitted for $K=1{,}000$), and $688$k and $22$k at $K=20{,}000$; these are upper bounds, since a node would not refill at that cost on a fixed trigger. With refills priced per amount it changes net revenue by at most $0.1$k (App.~\ref{app:skew}). Added to $1{,}000$\,ppm or
$L^{*}$ it changes net revenue by at most $0.6$k, within the noise of the gain above, at either cost: a measurable net gain needs refills on a fixed trigger at $K\ge1{,}000$\,sat.

\paragraph{Deployment.}
The level needs no node changes, only the ppm field; the skew needs a small program, written by the operator, that reads each channel's balance and posts the skew~\eqref{eq:trim}.

\begin{algbox}
\textbf{The \textsc{Trim} rule at a node.} \emph{Input:} local balances and forwards; the operator's weight on revenue, success and welfare.\\
\emph{Choose:} a large node traces the frontier by posting
a few uniform levels, each long enough in own forwards to tell them apart, and picks $L$; a small node takes a model's or a similar hub's level.\\
\emph{Run:} post $L$ on every channel. Only if refills remain frequent and costly,
fit each busy channel's flows, compute $\hat s_L$, and post
the skew~\eqref{eq:trim} whenever the balance moves by $5\%$ of capacity.\\
\emph{Guard:} stop other fee tools; probe within $[L/2,2L]$; to hold a chosen level, revert a move that loses more than twice the spread of the last eight probe pairs (it does not find one, \S\ref{sec:ev-choose}); skew only channels with enough traffic.
\end{algbox}

\paragraph{Scope.} The rule is for one hub choosing its own level: a hub repricing alone gains less than the $40$ together (App.~\ref{app:alone}), and $L^{*}$ is not an equilibrium, since with every hub at $L^{*}$ most gain by cutting (App.~\ref{app:level}). The peak is the revenue end of the frontier, not a recommendation: senders pay for it (\S\ref{sec:conclusion}).

\section{Network Evaluation}
\label{sec:eval}

In a network the level depends on senders and other hubs, so we test the results of \S\ref{sec:single}--\S\ref{sec:trim} by simulation, with every other node at its deployed fee. We stated predictions for the level, the skew and the tools before the runs; those for the skew and the tools are confirmatory (\S\ref{sec:ev-skew}, \S\ref{sec:ev-tools}).

\paragraph{Setup.}\label{sec:ev-case} The snapshot is the Lightning network of 4 August 2026 from public
gossip~\cite{rgs2024,lightning-rfc}, $5{,}507$ nodes and $27{,}063$ channels
with their posted fees. Topology, fees and the tools' rules are real; only
the payments are synthetic, as no public dataset carries forwarding traces. Each
run routes $12{,}000$ payments (mean $41$k sat) with capacities proxied by the
maximum HTLC size, scaled to $30\%$ so that depletion occurs within a run.
Senders accept a route with probability $\exp(-\kappa p/v)$, $\kappa=200$, for
path fee $p$ and amount $v$, take the cheapest path blind to balances, and retry
up to $10$ routes. They remember failures in one of two ways: with \emph{shared memory},
all senders avoid a failed channel for $300$ payments; \emph{mission
control} is kept per sender, as in LND; most senders pay once or twice, so it learns mostly within a payment's retries ($500$ persistent senders exercise it, Fig.~\ref{fig:network}(b)). Channels start balanced. The $40$ highest-degree hubs (the largest) run a candidate policy;
a hub channel below $5\%$ of capacity is refilled to $50\%$ at $K=100$ sats. We
run \texttt{charge-lnd}, feeadjuster and LNDg, implemented from their source and
tuned on calibration seeds over absolute ranges and over multiples of deployed
rates. We also run uniform levels from $250$ to $8{,}000$\,ppm, among them $L^{*}$, the revenue peak of a calibration scan ($4{,}900$\,ppm; App.~\ref{app:evalx}), and the anchored
skew~\eqref{eq:trim}.

\paragraph{Frontier and tests.} Results are on $10$ disjoint evaluation seeds with paired
$95\%$ intervals. The \emph{frontier} is the per-seed upper envelope of the revenue--success points of the deployed fees and the seven uniform levels run ($250$--$8{,}000$\,ppm; Table~\ref{tab:mainc} shows three), piecewise linear in success; the confirmatory tests add the tools' centers ($2{,}450$, $6{,}125$\,ppm) and $12{,}000$\,ppm; choosing the levels on other seeds moves the tools' gaps by at most $0.8$ points (App.~\ref{app:evalx}). Two claims on the main setting are confirmatory. Both use one-sided paired $t$-tests over the seeds, with a Holm correction over the $12$ tests (three tools and three skew comparisons per routing model; family-wise $\alpha=0.05$). First, no ppm-tuned tool beats the frontier by more than $3\%$ ($p\le0.009$ after Holm). Second, the skew adds less than $3\%$ under mission control and $5\%$ with shared memory, whose seed noise is larger. Everything else, including every robustness setting, is exploratory (App.~\ref{app:evfull}).

\subsection{The level is first order}
\label{sec:ev-results}

\begin{finding}
\FI{Under both routing models the level traces a revenue--success
frontier.}
\KN From deployed fees to $L^{*}$, hub revenue rises from $7.0$k to $118.9$k
with shared memory and from $38.6$k to $130.1$k under mission control, the $40$ hubs repricing together. At $1{,}000$\,ppm revenue is $2.6\times$ the deployed fees' under mission control and $8.1\times$ with shared memory.
\end{finding}

Each step of the level trades payment success for revenue
(Table~\ref{tab:mainc}, Fig.~\ref{fig:network}(a)) and cuts refills and refused forwards (Fig.~\ref{fig:network}(c)): at $L^{*}$ their counts per run fall by $84$--$98\%$. As
Proposition~\ref{prop:hub}(ii) predicts, uniform levels from $1{,}000$ to $8{,}000$\,ppm keep the
hubs' refills per thousand forwards at $3.1$--$3.9$ (shared memory) and
$1.6$--$2.3$ (mission control), against $12.6$ and $3.9$ under the deployed,
non-uniform fees. Much of the cut thus comes from posting one level on every channel, and the rest from carrying fewer forwards. Net of refills at $K=100$ the ordering is unchanged: refills cost $0.2$--$1.4\%$ of revenue on a level, and turn the deployed fees' net revenue negative with shared memory. Nor do levels that differ across hubs or channel classes beat the uniform one: per-hub levels found by coordinate ascent and the best split of draining and filling channels earn no more than $L^{*}$ on the evaluation seeds (App.~\ref{app:perhub}).

\paragraph{Where the peak sits.}\label{sec:ev-peak}
The peak is set by the senders' acceptance rule and by how much of the hubs'
demand has no way around them. We route each payment at deployed fees around all $40$ hubs. Then $35\%$ of the hubs' forwarded volume has no route around them, $99\%$ of it through one hub and $88\%$ because its sender or receiver reaches only hubs (App.~\ref{app:evalx}). Above a few hundred ppm the hubs thus earn from captive demand alone, $L\,e^{-\kappa L}$ per hub hop, which peaks at $1/\kappa=5{,}000$\,ppm (Proposition~\ref{prop:hub}(iii)). Predictions frozen before the runs held for the $10$, $20$ and $80$ largest hubs (App.~\ref{app:evalx}). Posted fees sit below $L^{*}$ on $97\%$ of channel directions, so real senders are likely more price-sensitive than ours, or competition binds. The peak's location and the size of the
multiples thus belong to the simulated senders; what we claim beyond them is
the ordering.

\begin{table}[t]
\centering
\caption{Policies on the 2026 snapshot under shared memory (S) and mission control (M): gross revenue (k sats per run), payment success, and gap to the frontier at the policy's success (\%, $95\%$ interval; negative: below it); tools tuned over ppm ranges or multiples ($\times$).}
\label{tab:mainc}
\Description{Table of gross revenue, payment success and gap to the frontier for deployed fees, tuned tools, the skew and uniform levels, under two routing models.}
\footnotesize
\setlength{\tabcolsep}{2.2pt}
\begin{tabular}{@{}lrrrrrr@{}}
\toprule
& \multicolumn{2}{c}{Gross} & \multicolumn{2}{c}{Success} & \multicolumn{2}{c}{Gap (\%)} \\
Policy & S & M & S & M & S & M \\
\midrule
deployed fee & $7.0$ & $38.6$ & $0.75$ & $0.70$ & $0$ & $-39.5{\pm}1.9$ \\
\midrule
LNDg, ppm & $115.0$ & $130.3$ & $0.62$ & $0.55$ & $-3.2{\pm}3.5$ & $-1.5{\pm}2.3$ \\
\texttt{charge-lnd}, ppm & $106.9$ & $118.5$ & $0.63$ & $0.54$ & $-4.2{\pm}4.0$ & $-10.5{\pm}2.8$ \\
feeadjuster, ppm & $94.2$ & $113.3$ & $0.66$ & $0.60$ & $-4.0{\pm}5.4$ & $-7.6{\pm}2.5$ \\
\midrule
\texttt{charge-lnd}, $\times$ & $37.1$ & $55.8$ & $0.68$ & $0.63$ & $-52.7{\pm}4.2$ & $-50.2{\pm}1.6$ \\
feeadjuster, $\times$ & $35.6$ & $57.7$ & $0.68$ & $0.63$ & $-57.5{\pm}3.4$ & $-48.8{\pm}2.0$ \\
\midrule
\ours skew at deployed & $6.8$ & $38.7$ & $0.75$ & $0.70$ & $-21.0{\pm}6.5$ & $-39.3{\pm}1.9$ \\
\ours level $1{,}000$ & $56.4$ & $100.9$ & $0.71$ & $0.66$ & -- & -- \\
\ours level $2{,}000$ & $88.1$ & $120.9$ & $0.68$ & $0.61$ & -- & -- \\
\ours level $L^{*}$ & $118.9$ & $130.1$ & $0.62$ & $0.54$ & -- & -- \\
\bottomrule
\end{tabular}
\end{table}

\subsection{The skew is second-order insurance}
\label{sec:ev-skew}

\begin{finding}
\FI{The skew changes gross revenue by at most a few percent, as Proposition~\ref{prop:value} allows here; it pays only where refills are frequent.}
\KN Alone at deployed fees or added to a level of $1{,}000$\,ppm or $L^{*}$,
the skew changes gross revenue within the noise under both routing models (below the confirmatory bounds); at deployed fees it cuts forced refills
by $21\%$ (shared memory) and $5\%$ (mission control), raising net revenue by $3.5\pm0.7$k and $0.2\pm0.1$k per run at $K=100$. On a level the hubs
refill only $2$--$9$ times per run, and the skew has little left to save.
\end{finding}

At deployed fees its $-21\%$ gap in Table~\ref{tab:mainc} comes from less success on the frontier's steep low end, not from less revenue. The skew has room to matter here: forwards are not small against the scaled capacities, as the size-biased mean forward $m_2/m_1$ is $6\%$ of capacity at the volume-weighted median hub channel, so Proposition~\ref{prop:value} allows a skew a few percent of revenue, as observed. At deployed fees, scaling the skew by $0.5$--$8\times$ keeps its refill cut at $19$--$23\%$ (shared memory) and $6$--$8\%$ (mission control) and gross revenue within $1.4\%$ of the deployed fees' ($5$ seeds; App.~\ref{app:amp}).

\subsection{Choosing the level explicitly}
\label{sec:ev-choose}

The forward-counting search of \S\ref{sec:trim} finds the level when the $40$ hubs pool their forwards: probing for $1{,}500$ payments each, it ends at $4{,}300$--$5{,}600$\,ppm (shared memory) and $3{,}300$--$4{,}000$\,ppm (mission control) from starts of $350$--$8{,}000$\,ppm, within $24{,}000$ payments. Per hub it overshoots ($8{,}700$--$10{,}500$\,ppm), and the probes cost $4$--$29\%$ of revenue against posting $L^{*}$ throughout (App.~\ref{app:search}). Probes sized in network payments do not identify the level: at $600$ payments a probe holds only $1$--$6$ forwards of the average hub (App.~\ref{app:evalx}), so searches started at $1{,}000$\,ppm, per hub or with one level for all, end $4$--$10\%$ below the frontier at their success, near where they start; when all hubs search at once, the end depends on the start. A guard that undoes losing moves holds a level but does not find one: from $350$\,ppm it reaches only $400$--$450$\,ppm and earns $16$--$27\%$ less than the unguarded search.

\subsection{Tuned tools mainly choose a level}
\label{sec:ev-tools}

\begin{finding}
\FI{On the 2026 snapshot no tool lies more than a few percent above the frontier.}
\KN Centered near $L^{*}$, the tools come within $1.5$--$10.5\%$ of the
frontier at equal success and earn no more than the flat level at their
center, within the noise.
\end{finding}

Tuned over absolute ranges, each tool's best range posts
$2{,}450$--$6{,}125$\,ppm at a balanced channel: tuning a tool chooses a level. On \texttt{charge-lnd}'s calibration grid (six centers by three spreads), the center moves revenue up to $3.1\times$ and $1.5\times$ at a fixed spread. At the same effective level, the other spreads earn at most $1.8\%$ more and up to $46\%$ less (App.~\ref{app:evalx}).
LNDg, near-flat, earns $3.2\%$ (shared memory) and $1.5\%$ (mission control) less than the frontier at its own success, inside the noise; \texttt{charge-lnd} and feeadjuster, which
spread the fee over the balance, earn $4.0$--$10.5\%$ less. Tuned as multiples of deployed fees,
which keep the deployed spread, the tools lie $48.8$--$57.5\%$ below the frontier; both picks sit on the grid's top edge. Mechanisms that price imbalance, \textsc{Fence}~\cite{wang2024fence}, Di Stasi et al.~\cite{distasi2018routing} and the Merchant~\cite{vanengelshoven2021merchant}, calibrated to the deployed rate, earn less than a uniform level on both the 2026 and the 2019 snapshot (App.~\ref{app:2019}). Tuned instead to the success of $1{,}000$\,ppm, the tools lie $8.1$--$10.0\%$ below the frontier with shared memory and $-4.3$ to $+2.9\%$ off it under mission control (App.~\ref{app:succ}).

\begin{table}[t]
\centering
\caption{From the exact model to the network, per parameter. Exact: single-channel medians; skew adds: best balance-dependent over best flat fee; level error: flat fee at half or twice the best level. MC: mission control. $^{\dagger}$Upper bound: best level on the evaluation seeds.}
\label{tab:bridge}
\Description{Table mapping each model parameter to its effect in the exact instances and in the network settings.}
\footnotesize
\setlength{\tabcolsep}{2pt}\renewcommand{\arraystretch}{0.95}
\begin{tabular}{@{}>{\raggedright\arraybackslash}p{0.085\textwidth}>{\raggedright\arraybackslash}p{0.17\textwidth}>{\raggedright\arraybackslash}p{0.19\textwidth}@{}}
\toprule
Parameter & Exact model & Network \\
\midrule
$\kappa$ & skew adds $1.0$--$2.2\%$ ($\kappa=5$--$20$) & $\kappa=100$--$2{,}000$: $L^{*}$ from $8{,}000$ to $500$\,ppm, as $1/\kappa$ \\
size / capacity & skew adds $0.8\to2.7\%$ (size $0.02\to0.2$ of capacity) & capacity $15$--$100\%$, payments $5\times$: multiples $3.3$--$20.9\times^{\dagger}$; tools at most $2.9\%$ above at full capacities \\
demand / replenishment & skew adds $0.0\to3.2\%$; level error $22\to50\%$ (medians, ratio $1.5\to8$) & uniform levels: $3.1$--$3.9$ refills per $1{,}000$ forwards, deployed fees $12.6$ \\
refill cost & skew adds $1.3\to2.2\%$; insures a low level only where refills are costly & $K=100$, $1{,}000$, $20{,}000$, economic refills: the skew's net gain scales with $K$ \\
hub coupling & draining hubs: per-channel levels cut the loss from $22$--$42$ to $11$--$13\%$, a skew on the uniform level does not & per-hub or per-class levels: no gain ($-7.8\%$ per hub, MC) \\
senders, routing & -- & two memory models (multiples $3.4$ vs $17.1\times$), memory pooled by $1$--$100$ wallets ($3.2$--$3.5\times$), wallet limits, persistent senders, heavier tails; 2019 balance-aware senders: tools $+14$--$24\%$ \\
\bottomrule
\end{tabular}
\end{table}

\paragraph{Across settings.}\label{sec:ev-robust} We repeat the test in the $19$ settings of Fig.~\ref{fig:network}(b) (Table~\ref{tab:bridge}) and on the 2019 snapshot. The ordering holds in every 2026 setting: in every row of Fig.~\ref{fig:network}(b) the best tool's estimate is at most $+2.9\%$. Its interval excludes $0$ in three cases, all under mission control: wallet fee limits ($+0.9\pm0.4\%$), LNDg at full capacities ($+2.8\pm2.0\%$; App.~\ref{app:x26}) and, outside Fig.~\ref{fig:network}(b), the success-constrained picks ($+2.7$--$2.9\%$). Tools win clearly only on 2019, with balance-aware senders, and not through a skew (\S\ref{sec:conclusion}).

\section{Limitations and Conclusion}
\label{sec:conclusion}\label{sec:discussion}

\paragraph{Limitations.} \emph{Traffic}: payments are synthetic, acceptance is exponential and demand is stationary. The peak's location reflects the simulated senders (\S\ref{sec:ev-peak}); the ordering holds across settings (Table~\ref{tab:bridge}).
\emph{Routing}: senders take the cheapest single path, blind to balances except through mission control. Scaled capacities put network success ($0.70$--$0.75$ at deployed fees) below
the $99.7\%$ a large node reports~\cite{river2023}, and closer to full capacities tools edge above the frontier (Fig.~\ref{fig:network}(b)).
The frontier is an empirical envelope of uniform levels, not the optimum over all policies. The tools are our reimplementations of their published rules, tuned on our seeds; on a live node's traffic they could do better. \emph{Model}: with an inbound fee the skew adds more, $3.8$--$8.3\%$, still below what a $20$--$25\%$ level error costs (App.~\ref{app:single}); multi-part payments would shrink forwards, tightening the bound of \S\ref{sec:sg-value}.
\emph{The 2019 exception}: on the 2019 snapshot the level still earns $3.6$--$15.5\times$ the deployed fees' revenue (App.~\ref{app:2019}), but with $\kappa=100$ and senders that see balances, tuned
\texttt{charge-lnd} earns $14$--$24\%$ more than the frontier at its success. So do the deployed fees at their best uniform multiple on the calibration seeds ($12\times$), which have no balance term and earn within $0.9\%$ of it: the gain is price discrimination across channels, the case of Proposition~\ref{prop:hub}(iv), not a skew; even there a skew on a uniform level changes revenue by $-1.7$ to $+2.6\%$.

\paragraph{Welfare.} Senders pay: from deployed fees, $L^{*}$ lowers total surplus by $7.3\pm1.3\%$ and $6.7\pm0.9\%$, against $2.0\pm0.6\%$ and $0.5\pm0.3\%$ at $1{,}000$\,ppm; at $L^{*}$ the surplus lost is $0.78$ and $0.99$ of the hubs' net gain. Counting a refused sender as a total loss, these are upper bounds (App.~\ref{app:welfare}).

\paragraph{Conclusion.} A routing node's level carries most of its revenue and should be chosen on a revenue--success frontier; the skew's value lies in fewer refills, where refills are costly. Tuned tools beat a flat level at equal success by at most a few percent, except with 2019's balance-aware senders. How real senders respond, and where levels settle when all hubs reprice, remain open.

\noindent\textbf{Use of generative AI.} An LLM assistant helped write code and edit text; the authors take full responsibility.

\clearpage
\bibliographystyle{ACM-Reference-Format}
\bibliography{trim}

\clearpage
\appendix

\section{Proofs and Exact Solutions}
\label{app:single}

\subsection{Proofs for Section~\ref{sec:single}}
We prove the results of \S\ref{sec:single} in their order. One more result is needed first: the diffusion fee curve and its slope at flow balance, which \S\ref{sec:sg-closed} cites and Fig.~\ref{fig:theory}(a) compares with the exact optimum.

\begin{proposition}[Fee curve and skew slope; diffusion]
\label{prop:skew}
In the diffusion regime with constant variance rate $\nu_0$ and exponential
demand, let $g$ be the optimal average reward, taken from the exact jump model, and assume $g<F_{\min}:=\min_{s\ge0}F(s)$, which is $R^{*}$ if
$\Lambda m_1\ge e\,\mu m_1^{+}$ and $\Phi(0)$ otherwise (Proposition~\ref{prop:ceiling} gives
$g\le R^{*}$ when $k\ge\hat s^{*}$). Then:
\emph{(i)} (cf.~\cite[eqs.~(20)--(22)]{ata2005drift}) above the refill trigger the optimal shadow price
solves
\begin{equation}
  \hat s'(b)=-\tfrac{2}{\nu_0}\big(F(\hat s(b))-g\big),\qquad \hat s(c)=0,
  \label{eq:ode}
\end{equation}
and the optimal
fee $f^{*}(b)=\hat s(b)+1/\kappa$ is strictly decreasing in $b$;
\emph{(ii)} if $\Lambda m_1>e\,\mu m_1^{+}$ and the curve reaches $\hat s^{*}$,
then at the flow-balance point $b^{*}$, where $\hat s(b^{*})=\hat s^{*}$ and the
fee is $\bar f$, the skew slope is
\begin{equation}
  a:=-\hat s'(b^{*})=\frac{2\,(R^{*}-g)}{\nu_0},
  \label{eq:slope}
\end{equation}
and the skew is at least this steep everywhere.
\end{proposition}

\begin{proof}[Proof of Proposition~\ref{prop:single-opt}]
Between events the balance does not move, and every fee gives total event rate
$\lambda(f)+\mu\le q:=\lambda(0)+\mu$ (Assumption~\ref{ass:demand}). Uniformizing at
rate $q$, the continuation operator
$(T w)(b)=\sup_{f\ge0}\big[(\mathcal L^{f}w)(b)+q\,w(b)\big]/(\rho+q)$ has bounded
rewards ($f\lambda(f)$ is continuous and vanishes as $f\to\infty$; $\xi\le c$; $\sup\emptyset=-\infty$ in $\mathcal M$ at $b=c$) and is a contraction with modulus
$q/(\rho+q)$ on bounded functions. Two consecutive refills cost at least as much as
one of the same total size ($K>0$, linear $k$), so refilling twice in a row is never
optimal, and $V$ is a fixed point of $w\mapsto\max\{Tw,\mathcal MTw\}$. $\mathcal M$
is nonexpansive (a supremum of translates less constants), so this map is again a
contraction with modulus $q/(\rho+q)$, and its unique bounded fixed point is $V$.
Rearranged, the fixed-point equation is \eqref{eq:qvi}, and any bounded solution of
\eqref{eq:qvi} is such a fixed point~\cite{davis1993markov,oksendal2019impulse}.
\end{proof}

\begin{proof}[Proof of Proposition~\ref{prop:fee-struct}]
The $f$-dependent term of \eqref{eq:gen} is
$u(f)=\lambda(f)[f\,m(b)+D(b)]$ with
$D(b)=\int_0^b(V(b-\xi)-V(b))G(d\xi)-c_{\mathrm r}\bar G(b)=-\hat s(b)m(b)$.
So $u=m(b)\lambda(f)(f-\hat s(b))$ is continuous, positive exactly for $f>\hat s(b)$ (or everywhere if $\hat s(b)<0$) and vanishes
as $f\to\infty$ (Assumption~\ref{ass:demand}): a maximizer exists, and at a positive one $u'(f)=0$; dividing by
$-\lambda'(f)m(b)>0$ gives $f-\hat s(b)=\psi(f)$. Where $u>0$,
$\log u=\log m(b)+\log\lambda(f)+\log(f-s)$ with $s=\hat s(b)$ has
$\partial^2/\partial f\partial s=(f-s)^{-2}>0$. If $s<s'$, $f$ maximizes at $s$, $f'$ at $s'$ and $f'<f$, then
$f>f'>s'$ and $\log u(f,s')-\log u(f',s')>\log u(f,s)-\log u(f',s)\ge0$, so $f'$ is not a maximizer: every maximizer at $s'$
is at least every maximizer at $s$, and the largest is nondecreasing. If $f-\psi(f)$ is strictly increasing the root is unique.
For exponential demand $\psi\equiv1/\kappa$.
\end{proof}

\begin{lemma}[Monotone, concave value]
\label{lem:concave}
In the diffusion regime of \S\ref{sec:sg-closed}, $V$ is nondecreasing and
concave and $\hat s$ is nonincreasing, strictly under the average criterion
(Proposition~\ref{prop:skew}).
\end{lemma}

\begin{corollary}[Reservation-fee skew]
\label{cor:skew-mono}
Wherever $\hat s$ is nonincreasing in the balance, so is $f^{*}$
(Proposition~\ref{prop:fee-struct}): the node charges more on a depleted channel
and less on a flush one.
\end{corollary}

\begin{proof}[Proof of Lemma~\ref{lem:concave}]
In the diffusion regime \eqref{eq:hjb-diff} with discounting,
$s=V'$ solves $\rho s=[\Phi'(s)+\mu m_1^{+}]s'+\tfrac{\nu_0}{2}s''$ with
$s=k\ge0$ at the refill boundary (or, without refills, $s(0)\ge0$: $V$ is nondecreasing, since a copy started with more balance, under the same fees, refuses no more forwards until the two meet) and $s(c)=0$.
\emph{Monotonicity.} At an interior minimum, $s'=0$ and $s''\ge0$ force
$\rho s\ge0$; with the boundary values, $s\ge0$.
\emph{Concavity.} At an interior local maximum $s''\le0$ forces $s\le0$, so
$s=0$ there. If $s'>0$ somewhere, then, since $s\ge0$ and $s(c)=0$, $s$ has a
later interior local maximum with a positive value, a contradiction; hence
$s'\le0$ and $V$ is concave; in the diffusion regime $\hat s$ is identified with $s$, so it is nonincreasing.
Strict decrease under the average criterion is Proposition~\ref{prop:skew}(i), under its hypotheses.
\end{proof}

\begin{proof}[Proof of Proposition~\ref{prop:ceiling}]
Fix $s\in[0,k]$. By the definition of $\Phi$, $m_1\lambda(f)(f-s)\le\Phi(s)$ for every
fee $f$. Under any policy let $Q$ and $U$ be the long-run average rates (along horizons attaining the $\liminf$ in $g$) of served
demand volume and of refilled volume; the balance being bounded, $Q\le\mu m_1^{+}+U$.
At fee $f$ volume is served at a rate $q\le m_1\lambda(f)$, and
$fq=(f-s)q+sq\le\Phi(s)+sq$ (the first term is negative if $f<s$, and $\Phi\ge0$), so fee
income is at most $\Phi(s)+sQ$. Refills cost at least $kU$ and refusal penalties are
nonnegative, so $g\le\Phi(s)+s\,\mu m_1^{+}+(s-k)U\le F(s)$. $F$ is convex, $\Phi$ being a supremum of functions affine in $s$.
Since $\Phi(s)\ge m_1\lambda(\bar f)(\bar f-s)$, $F(s)\ge R^{*}$ for every $s$, with equality at an $\hat s^{*}$ where $\bar f$ attains the supremum (assumed: with demand that is not log-concave the maximizer can jump over $\bar f$); there $F(\hat s^{*})=m_1\lambda(\bar f)(\bar f-\hat s^{*})+\mu m_1^{+}\hat s^{*}=\mu m_1^{+}\bar f=R^{*}$.
$\hat s^{*}\ge0$ when the fee maximizing $\lambda(f)f$ is at most $\bar f$, and $\hat s^{*}\in[0,k]$ when also $k\ge\hat s^{*}$.
\end{proof}

\begin{proof}[Proof of Proposition~\ref{prop:skew}]
\emph{Diffusion reduction.} Expanding $h(b)-h(b-\xi)=\xi h'(b)-\tfrac12\xi^2h''(b)+O(\xi^3)$
in \eqref{eq:shadow} (with $h$ in place of $V$),
\[
  \hat s(b)=h'(b)-\tfrac12h''(b)\tfrac{m_2}{m_1}+O\Big(\tfrac{m_3}{m_1}\Big)+R(b),
\]
where $R(b)$ collects the refusal term and the truncation of the integral
at $b$; both are bounded by a constant times $\bar G(b)/m(b)$, which is
$O(\varepsilon)$ in the region considered. Since $|h''|/|h'|=O(1/\chi)$,
$\hat s=h'+O(\varepsilon)$. Expanding the replenishment term the same way and
maximizing the demand term over $f$ turns the average-reward optimality
equation into \eqref{eq:hjb-diff}.
\emph{The curve.} Put $s=h'$. Then \eqref{eq:hjb-diff} reads
$g=F(s)+\tfrac{\nu_0}{2}s'$, which is \eqref{eq:ode}; $s(c)=0$ because
refused replenishment at $b=c$ is costless (reflection without cost). This is the
reduction of~\cite[\S3]{ata2005drift}: with drift $d(f)=\mu m_1^{+}-m_1\lambda(f)$,
reflected coordinate $z=c-b$, and cost rate $\frac{m_1\Lambda}{\kappa e}-m_1\lambda(f)f$
(restricting $f\ge1/\kappa$, which is without loss since $s\ge0$), their
$\varphi$ and $\gamma$ become $F-\frac{m_1\Lambda}{\kappa e}$ and
$\frac{m_1\Lambda}{\kappa e}-g$, so their eqs.~(20)--(22) are \eqref{eq:ode}
and the quadrature. The equation holds for the optimal $g$, which the lower
(depletion or refill) boundary pins; in the diffusion limit the fee at $b=0$
is a degenerate control, so we take $g$ from the jump model. For exponential
demand the supremum in $\Phi$ is attained at $f=s+1/\kappa$, so
$f^{*}=\hat s+1/\kappa$ and $\Phi(s)=\frac{m_1\Lambda}{\kappa e}e^{-\kappa s}$.
Then $F'(s)=\mu m_1^{+}-m_1\lambda(s+1/\kappa)$ and $F''(s)=\kappa m_1\lambda(s+1/\kappa)>0$,
so $F$ is strictly convex with its minimum where $m_1\lambda(s+1/\kappa)=\mu m_1^{+}$,
i.e.\ at $\hat s^{*}$, and $F(\hat s^{*})=\frac1\kappa\mu m_1^{+}+\mu m_1^{+}\hat s^{*}=R^{*}$.
\emph{(i)} Since $F(s)\ge F_{\min}>g$ on $s\ge0$, \eqref{eq:ode} gives $s'<0$ everywhere.
\emph{(ii)} At $b^{*}$, $F(s(b^{*}))=R^{*}$, which gives \eqref{eq:slope}; $b^{*}$ is
the flow-balance point because the controlled drift
$\mu m_1^{+}-m_1\lambda(f^{*}(b))$ vanishes exactly where $s=\hat s^{*}$; and
$-s'(b)=\frac2{\nu_0}(F(s(b))-g)\ge\frac2{\nu_0}(R^{*}-g)=a$.
\emph{Remark.} If $\Lambda m_1\le e\,\mu m_1^{+}$, then for $s\ge0$,
$F'(s)=\mu m_1^{+}-m_1\lambda(s+1/\kappa)\ge\mu m_1^{+}-m_1\Lambda/e\ge0$, and
under the optimal fee $f^{*}\ge1/\kappa$ the drift
$\mu m_1^{+}-m_1\lambda(f^{*})\ge\mu m_1^{+}-m_1\Lambda/e\ge0$, so the balance
drifts toward $c$, where $s(c)=0$ and $f^{*}=1/\kappa$: the fee there is the markup alone.
With discounting, $\rho V$ replaces $g$ in \eqref{eq:hjb-diff} and the same steps
give $a=2(R^{*}-\rho V(b^{*}))/\nu_0$.
Integrating \eqref{eq:ode} gives $b(s)=c-\frac{\nu_0}{2}\int_0^{s}\frac{dv}{F(v)-g}$.
\end{proof}

\begin{proof}[Proof of Proposition~\ref{prop:value}]
\emph{(i)} By Proposition~\ref{prop:ceiling}, $g\le R^{*}$.
Equality $g=R^{*}$ would need the fee to equal $\bar f$ almost
surely (the supremum defining $\Phi(\hat s^{*})$ is attained only there), and that
constant fee refuses demand at a positive rate (below), so $g<R^{*}$. Under the
constant fee $\bar f$ without refills, let both sides' sizes have law $G$; then
$m_1=m_1^{+}$ and flow balance give $\lambda(\bar f)=\mu$, and the uniform law on $[0,c]$
satisfies detailed balance: a jump from $b$ to $b-x$, allowed iff $x\le b$, has rate
$\lambda(\bar f)G(dx)$, and the reverse jump from $b-x$ to $b$, allowed iff $b\le c$, has rate
$\mu G(dx)$. Under it the refused demand volume and count per unit time are
$\lambda(\bar f)c^{-1}\!\int_0^c\mathbb E[\xi\mathbf 1\{\xi>b\}]\,db=\lambda(\bar f)\mathbb E[\xi(\xi\wedge c)]/c$ and
$\lambda(\bar f)\mathbb E[\xi\wedge c]/c$; for $\xi\le c$, with $\nu_0=2\lambda(\bar f)m_2$, these are
$\nu_0/(2c)$ and $(m_1/m_2)\,\nu_0/(2c)$. For any size laws, in the diffusion regime the balance is a driftless
Brownian motion with variance rate $\nu_0$ reflected at $0$ and $c$, again uniform, with refused volume
rate $\tfrac{\nu_0}{2}\cdot\tfrac1c$ (the local-time rate at $0$). Each unit of refused volume costs the
fee $\bar f$ and each refused forward $c_{\mathrm r}$, which gives the stated revenue. Subtracting gives the gain bound,
and \eqref{eq:slope} with $g\ge R^{*}-(\bar f+c_{\mathrm r}m_1/m_2)\nu_0/(2c)$ gives
the slope bound.
\emph{(ii)} With drift $d<0$ the local-time rate at $0$ is $|d|/(1-e^{-2|d|c/\nu_0})=|d|(1+o(1))$ when $|d|c\gg\nu_0$; under the affine skew the balance is an Ornstein--Uhlenbeck process around $b_{\mathrm{eq}}$ with variance $\nu_0/2\theta$, whose boundary densities are Gaussian tails.
The exponent holds while $\kappa a\min(b_{\mathrm{eq}},c-b_{\mathrm{eq}})\lesssim1$;
beyond, the drift toward depletion saturates at $\mu m_1^{+}$ and the exponent
grows linearly, as $2\mu m_1^{+}b_{\mathrm{eq}}/\nu_0$.
\end{proof}

\subsection{Exact single-channel solutions}
We solve the jump-model optimality equation exactly on $18$ instances, $\kappa\in\{5,10,20\}$, two demand-to-replenishment ratios and three refill costs. On all $18$ instances $h$ is nondecreasing and concave except on at most $5\%$ of grid points, at the refill kink and the boundaries; $f^{*}$ is nonincreasing except on one instance, on $7.5\%$ of grid steps just above the refill trigger.
The $16$ of the $18$ with an interior flow-balance point satisfy the slope identity
\eqref{eq:slope} within $0.12\%$ ($0.13\%$ with the $8$ binding grid instances whose flow-balance point lies well above the trigger). Table~\ref{tab:t1} (App.~\ref{app:t1}, $K=0.05$) gives the jump model's loss under the ODE curve
\eqref{eq:ode}, the best constant fee (with and without $c_{\mathrm r}$) and $1/\kappa$. Given the exact $g$, on the $6$
instances of that table the curve \eqref{eq:ode} loses at most
$0.3\%$ of optimal revenue, against $0.7$--$5.4\%$ for the best constant fee ($0.7$--$2.4\%$ without $c_{\mathrm r}$). The grid of \S\ref{sec:sg-exact} has $144$ instances ($\kappa\in\{5,10,20\}$, demand-to-replenishment ratios $1.5$--$8$, forwards of $2$--$20\%$ of capacity, cheap and costly refills, with and without $c_{\mathrm r}$). The anchored skew~\eqref{eq:trim} at $1/2$, $1$ and $2$ times the best constant fee takes $\hat s_L$ from the flat fee's value function and $\bar b$ under the posted fee (\texttt{t1\_skew\_insurance.py}); the $28$ instances where the half level drains the channel into a net loss all have $c_{\mathrm r}>0$.

\paragraph{Small forwards and estimation error.} With forwards scanned from $16$ to $1\%$ of capacity ($\kappa=10$, demand $3$ or $8$ times replenishment, with and without $c_{\mathrm r}$, no refills; \texttt{t1\_asymptotic.py}), the balance under a flat fee is a reflected diffusion whose boundary local times give the best flat fee's loss $\mu m_1^{+}\delta\min_x[\bar f\vartheta(x)+c_{\mathrm r}'\vartheta(-x)-x/\kappa]+O(\delta^2)$, with $\delta=\nu_0/(2c\mu m_1^{+})$ the share of volume refused, $\vartheta(x)=x/(1-e^{-x})$ and $c_{\mathrm r}'=c_{\mathrm r}m_1/m_2$ ($x=0$ gives $\Gamma$); with demand eight times replenishment the exact loss is within $2.2\%$ of it at every size, and within $0.4\%$ of $\Gamma$ without $c_{\mathrm r}$. In the small-forward (wide-channel) limit the optimum loses $2\pi^2\mu m_1^{+}\delta^2/\kappa$, from the Riccati form of \eqref{eq:hjb-diff}: the exact ratio rises from $0.33$ to $0.91$ as forwards shrink, and the gain reaches $0.82$ of $\Gamma$ at $1\%$. With demand three times replenishment, nearer the markup regime, convergence is slow: the first-order term is $1.3$--$2.3$ times too small and the $\delta^2$ term reaches $0.15$--$0.42$ of the exact loss. Scaling the anchored skew's shadow price by $\beta\in\{0.5,0.8,1.25\}$ on ten instances (\texttt{t1\_skew\_robust.py}) changes the loss at the best flat fee by at most $1.1$ points; at half that fee, where six of them drain into a net loss, any $\beta>0$ cuts that loss (from $101$--$262$ to $57$--$122\%$).

\paragraph{Other size laws and demand curves.} The $2{,}592$ instances (\texttt{t1\_general.py}) cross exponential, lognormal and Pareto forward sizes with exponential, mixed (captive plus competed, $w=0.44$) and linear demand, six demand-to-replenishment ratios, four forward sizes, four refill costs and three refusal penalties; the flat-fee gap is $1.1$--$2.3\%$ at the median and at most $15.8$, $24.1$ and $21.5\%$ by demand. On the $684$ with $\bar f$ above the fee that maximizes $\lambda(f)f$ and $k\ge\hat s^{*}$, the optimum never exceeds $R^{*}$ (at most $0.999R^{*}$).

\paragraph{Inbound fees.} When the node also posts an inbound fee $g_{\mathrm{in}}$ on replenishment, accepted with probability $e^{-\kappa g_{\mathrm{in}}}$ (\texttt{t1\_inbound.py}; six instances, forwards of $8\%$ of capacity), its best inbound fee is $1/\kappa$ less the marginal value of incoming balance, so the skew enters both fees with opposite signs. A balance-dependent pair adds $3.8$--$8.3\%$ over the best flat pair ($0.1$--$7.2\%$ for the outbound fee alone), while a flat fee $20\%$ below or $25\%$ above the best one, with the inbound fee re-optimized, loses $10.8$--$18.8\%$. At a hub the entering forward also pays an outbound fee elsewhere, so we report the structure, not the inbound income.

\subsection{The hub}
\paragraph{Proof of Proposition~\ref{prop:hub}.}
(i) Channel $i$'s generator \eqref{eq:gen} involves the others only through
$\mu_i=\sum_{j\ne i}\pi_{ij}\lambda_j(L+s_j(b_j))$. Its time average
$\mu_i(L)=\sum_{j\ne i}\pi_{ij}\Lambda_j\ell(L)\bar\ell_j$, with $\bar\ell_j$ the
mean of $\ell(L+s_j)/\ell(L)$ ($1$ at the anchor of \eqref{eq:trim}), decouples
the channels into problems of \S\ref{sec:single}; Proposition~\ref{prop:fee-struct} gives each skew.
(ii) Under flat fees $\bar\ell_j=1$, so $\mu_i(L)/\lambda_i(L)=\sum_{j\ne i}\pi_{ij}\Lambda_j/\Lambda_i$,
free of $L$, so $d_i(L)=\lambda_i(L)\big(m_1^{+}\mu_i/\lambda_i-m_1\big)$ has a
sign free of $L$ and scales with the forward rate.
(iii) $g_i=L\,m_1\lambda_i(L)(1-p_i)$ less refusal and refill costs, with
$p_i$ the refused share; where these are rare,
$g_i=L\,m_1\Lambda_i\ell(L)+o(1)$, and the sum's maximizer is free of the
$\Lambda_i$. For $\ell=e^{-\kappa f}$, $L\ell(L)$ peaks at $1/\kappa$ and
$rL^{*}\ell(rL^{*})/L^{*}\ell(L^{*})=re^{1-r}$.
(iv) By Proposition~\ref{prop:skew} the level is $\max\{L^{*},\bar f_i\}$ with
$m_1\lambda_i(\bar f_i)=m_1^{+}\mu_i$. Raising $f_j$ lowers $\mu_i$ and so raises $\bar f_i$:
the map $f\mapsto(\min\{\max\{L^{*},\bar f_i(f_{-i})\},f_{\max}\})_i$ is nondecreasing on the
complete lattice $[0,f_{\max}]^n$, with $f_{\max}$ a fee above which demand is negligible, so
Tarski's theorem gives a fixed point; with
shapes $\ell_j$ that differ, so does the markup $\arg\max_L L\ell_j(L)$.
(v) Under the decoupling of (i), it sums Proposition~\ref{prop:value}(i) over the channels. \qed

\paragraph{Exactly solved hubs.} The hub of \S\ref{sec:sg-hub} has $n=2$ modelled channels (capacity $1$, a grid of $40$ balance steps each) and a deep channel $R$, the rest of the hub, with no liquidity limit (\texttt{t1\_hub\_exact.py}). A forward from $i$ to $o$ ($i\ne o$ in $\{1,2,R\}$) arrives at potential rate $\Lambda_{io}$, is accepted with probability $e^{-\kappa_o f_o}$, has an exponential size of mean $m_1\in\{0.05,0.1\}$ on the grid, needs its size on the peer's side of $i$ and on the hub's side of $o$ (else it is refused, at cost $c_{\mathrm r}=0.01$), and moves its size from $b_o$ to $b_i$. Rates, costs and solver details are in App.~\ref{app:hubx}, the grid in Table~\ref{tab:hub}.

\begin{table}[H]
\centering
\caption{Exactly solved two-channel hubs ($48$ instances: forwards of $5$ and $10\%$ of capacity, drain ratios $1$, $1.5$ and $3$, two traffic shares, equal or different demand shapes, two refill costs): loss against the exact optimum (\%, median and range; above $100\%$ a net loss).}
\label{tab:hub}
\Description{Table of the losses of four policies against the exact optimum on two-channel hubs, by flow balance and demand shape.}
\footnotesize
\setlength{\tabcolsep}{1.6pt}
\begin{tabular}{@{}lrrrr@{}}
\toprule
Hub & Uniform & Per-ch. & Unif.\,$+$\,skew & Per-ch.\,$+$\,skew \\
\midrule
bal., equal shapes (8) & $12.6$ \scriptsize{$11$--$17$} & $11.4$ \scriptsize{$11$--$16$} & $5.8$ \scriptsize{$3$--$12$} & $3.9$ \scriptsize{$3$--$7$} \\
drain., equal shapes (16) & $42.3$ \scriptsize{$21$--$84$} & $11.0$ \scriptsize{$8$--$16$} & $39.7$ \scriptsize{$19$--$163$} & $4.8$ \scriptsize{$3$--$8$} \\
bal., shapes differ (8) & $31.6$ \scriptsize{$27$--$36$} & $12.7$ \scriptsize{$12$--$17$} & $27.8$ \scriptsize{$25$--$31$} & $4.2$ \scriptsize{$3$--$6$} \\
drain., shapes differ (16) & $22.7$ \scriptsize{$17$--$34$} & $12.0$ \scriptsize{$11$--$16$} & $19.8$ \scriptsize{$11$--$30$} & $4.8$ \scriptsize{$4$--$8$} \\
\bottomrule
\end{tabular}
\end{table}

\section{Evaluation Design}
\label{app:evalx}

\paragraph{Parameters and code.}
Distributions, seeds, grids, tool rules and solver settings are in App.~\ref{app:params}; code, raw outputs and a number checker will be made public upon acceptance; null tests check that each extension reproduces the base simulator~\cite{lnfeeempirical} bit for bit.

\paragraph{Traffic and calibration.} Offered focal demand falls exponentially in the posted fee near the deployed level ($R^2=0.95$ over four exogenous fee scalings). A policy's gap interpolates the frontier per seed at the policy's own success. $L^{*}$ comes from a calibration scan of
$15$ uniform levels ($25$--$12{,}000$\,ppm) with balances frozen, never on the grid edge in any setting. The captive share is computed from the snapshot (\texttt{analysis/param\_topo.py}). The captive-demand prediction of \S\ref{sec:ev-peak} peaks at $4{,}990$\,ppm, stays within $0.02$ of the scan's revenue relative to its peak at every level, and moves to $1/\kappa$ for $\kappa$ from $100$ to $2{,}000$, where the evaluated best levels lie at $0.6$--$1.0$ times it. The demand mixture of \S\ref{sec:model-informal} is a two-part fit $\lambda(L)=\Lambda[(1-w)e^{-(\kappa_c+200)L}+we^{-200L}]$ to the calibration scan, with captive share $w=0.44$. The frozen predictions for $10$, $20$ and $80$ hubs and their test are in App.~\ref{app:stats}. With $160$ and $320$ hubs repricing ($5$ seeds, not predicted in advance) the best evaluated level stays at $3{,}500$--$4{,}900$\,ppm, earning $15.5$ and $19.1\times$ the deployed fees' revenue with shared memory and $4.3$ and $5.3\times$ under mission control, while network success at it falls to $0.48$ and $0.45$ with $320$ hubs (from $0.75$ and $0.70$ at deployed fees).

\paragraph{Tools.}
A ``day'' of LNDg is $300$ payments; with a day of $100$ or $30$ payments or its default cap, LNDg lies $1.1$--$16.6\%$ below the frontier, and \texttt{charge-lnd} re-evaluated on every balance change $7.3\%$ and $11.1\%$. We reimplement \texttt{charge-lnd}'s linear rule in the local share, feeadjuster's $50^{0.5-x}$ multiplier and LNDg's daily tiered adjustment. On a grid extended to $12{,}000$\,ppm feeadjuster's score keeps rising to the edge, but the pick there lies $13.1\pm6.6\%$ below the frontier on the evaluation seeds, against $4.0\pm5.4\%$ for the pick we report. Choosing the frontier's levels on the other seeds changes the tools' gaps by at most $0.8$ points (\texttt{loo\_frontier.py}). Paired by seed, no tool earns significantly more than the flat level at its center (per-tool intervals in App.~\ref{app:stats}).

\paragraph{Power.} With $10$ seeds, a tool tied with the frontier would be shown not to exceed it by more than $4$--$5\%$ (mission control) and $7$--$10\%$ (shared memory) with $80\%$ power at the strictest Holm level; the tuned tools lie $2$--$11\%$ below it (on a wider frontier that also includes $2{,}450$, $6{,}125$ and $12{,}000$\,ppm). Per-test power and probe sizes: App.~\ref{app:stats}.

\subsection{Full design and results}
\label{app:evfull}
The tuned picks are, with shared memory and mission control, LNDg starting at $4{,}900$\,ppm, \texttt{charge-lnd} over $(2{,}450,9{,}800)$\,ppm and feeadjuster soft at $4{,}900$ / $2{,}450$\,ppm, and as multiples \texttt{charge-lnd} $(16,1024)\times$ and feeadjuster default at $512\times$.

\paragraph{Ordering across settings.}\label{app:x26} Per-setting results for every row of Fig.~\ref{fig:network}(b) (full capacities with $10$ seeds, others as listed there) are in App.~\ref{app:settings}; Table~\ref{tab:bridge} summarizes them.

\paragraph{Refills and the skew.}\label{app:skew} With economic refills ($K=100$\,sat plus $300$\,ppm of the amount, only when the liquidity sold covers them), hubs at deployed fees refill under once per run, since refills rarely pay there, and the skew changes their net revenue by at most $0.1$k.

\paragraph{The 2019 snapshot.}\label{app:2019}
On the 2019 snapshot of Beres et al.~\cite{lnfeeempirical} (real capacities; $9{,}000$ payments, shared memory, $K\in\{100,1{,}000\}$) two-thirds of the channel policies post LND's default $1$\,ppm. Across hub tiers, destination skews, capacity scales and $\kappa\in\{100,400\}$, the level earns $3.6$--$15.5\times$ the deployed fees' revenue at $L^{*}=2{,}406$--$9{,}671$\,ppm, with $59$--$97\%$ fewer refills, and the skew cuts refills at deployed fees by $10$--$43\%$. \textsc{Fence}, the Merchant and Di Stasi et al., calibrated to each channel's deployed rate, earn $2.1$--$13.7$k on 2019, below the deployed fees' $20.3$k, and on 2026 lie $12$--$89\%$ below the frontier at their success. The tools' gains of \S\ref{sec:conclusion} come from runs with $\kappa=100$ and senders that route on known balances.

\section{The Level}
\label{app:level}

\paragraph{Levels in public gossip.}\label{app:gossip}
In the $336$ gossip snapshots of~\cite{valko2025geo} (2019--2023, posted fees) we follow the nodes ever among the $40$ largest. Their level moves $1.7$ times a year, by a median $2.3\times$, most channels at once. Pooled over fee setters (an operator's nodes counted as one), $50\%$ of the posted log-fee variation is the level over time, $33\%$ persistent differences across channels and $16\%$ the time-varying rest. Raises stick: after six months $83\%$ are kept or extended, against $60\%$ of cuts (Fisher $p=0.006$), and the hubs' median level rose from $1$\,ppm in 2019--20 to $149$--$197$\,ppm in mid-2023.

\paragraph{Competition.} When all reprice, $L^{*}$ is not an equilibrium: in an exploratory best-response check ($5$ hubs, $10$ seeds), with every hub at $L^{*}$ none gains by raising to $8{,}000$\,ppm and most gain by cutting, taking revenue from the other hubs, as the $40$ together earn no more.

\sloppy

\section{Exactly Solved Hubs: Rates, Solver and Three-Channel Hubs}
\label{app:hubx}

\paragraph{Exactly solved hubs, rates and costs.} With $g=(\rho^{1/2},\rho^{-1/2})$, $\Lambda_{io}=wFg_o/g_i$ between the modelled channels, $\Lambda_{Ro}=(1-w)Fg_o$ and $\Lambda_{iR}=(1-w)F/g_i$, $F=0.4/m_1$, $\rho\in\{1,1.5,3\}$ and $w\in\{0.5,0.8\}$; $\kappa=10$ on every exit, or $20$ on channel~$1$ when the shapes differ; refills top any set of channels up at cost $K+k\Delta$ per channel, $K\in\{0.05,0.5\}$, $k=0.05$.

\paragraph{Exactly solved hubs: solver and per-instance results.} Policy iteration with exact linear solves gives the optimum, every exit's fee set state by state in closed form; for a fixed fee rule the same iteration optimizes the refills. The best uniform level is found by a bounded scalar search, the per-channel levels (both channels and $R$) by coordinate ascent from it; the anchored skew uses each channel's shadow price averaged over the other balance under the flat policy's stationary law, anchored at the time-average balance. Part~(v) applies to the $24$ hubs with equal shapes. Per instance, in the $8$ two-channel hubs with balanced flows and equal shapes the channels' best levels differ from the uniform level by at most $2.3\%$ and add $0.1$--$2.0$ points; the $8$ that meet $k\ge\hat s^{*}$ reach $0.41$--$0.83$ of the bound of (v), and of the $16$ that do not, $2$ exceed it ($1.18$ and $1.51$). With three modelled channels ($4$ instances, forwards of $20\%$ of capacity, drain ratios $1$ and $3$) the cases repeat: with balanced flows the channels' own best levels match the uniform level and add at most $1.6$ points; with a draining channel the uniform level loses $33$--$39\%$, per-channel levels $9$--$12\%$, and a skew on the uniform level $36$--$67\%$; the bound of (v) holds in all four.

\paragraph{Exactly solved hubs, global check.} On three instances a global grid over the three levels ($7^3$ cells) followed by coordinate ascent reaches the same optimum as coordinate ascent from the best uniform level (\texttt{t1\_hub\_global.py}).

\section{Hubs Repricing Alone and Competition}
\label{app:alone}

\paragraph{A hub alone.} A hub repricing alone, all other nodes at deployed fees ($10$ seeds; the largest hub and, drawn at random, one per quartile of deployed revenue, two under mission control), gains less: under mission control $1{,}000$\,ppm raises the revenue of seven of nine hubs, by $1.2$--$2.9\times$ (one more from a $29$-sat base), and lowers two; $L^{*}$ raises four significantly and lowers three. With shared memory $1{,}000$\,ppm raises each of the five hubs tried, by $2.5\times$ or more.

\paragraph{Competition.} When all reprice,
$L^{*}$ is not an equilibrium. In an exploratory best-response check ($5$
hubs: the largest, $H$, and one per quartile of deployed revenue; $10$ seeds), with every hub at $L^{*}$ none gains by
raising to $8{,}000$\,ppm and most gain by cutting: $H$ most at $2{,}000$\,ppm of the levels tried ($+21\pm9\%$ shared memory, $+44\pm15\%$ mission control),
taken from the other hubs, as the $40$ together earn no more.

\section{Per-Hub and Per-Class Levels}
\label{app:perhub}

\paragraph{Per-hub and per-class levels.} Nor do levels that differ across hubs beat the uniform one. From $L^{*}$ we ran coordinate ascent on the calibration seeds, halving or doubling each hub's level. With shared memory it kept none of the $80$ moves, as none gained more than twice its paired standard error, nor any of the $8$ moves of four degree strata (App.~\ref{app:evalx}). Under mission control it kept seven, whose separate gains of $8.3\%$ shrank to $0.7\%$ together, as the hubs took traffic from each other; on the evaluation seeds the seven lost $7.8\pm4.4\%$ of revenue against $L^{*}$. Nor does a level per channel class, which the exact hubs of \S\ref{sec:hub} favor. We class each hub's channels as draining ($6$--$8\%$ of them) or filling by net flow at $L^{*}$. The best split over $\{L^{*}/2,L^{*},2L^{*}\}^2$ on the calibration seeds puts $2L^{*}$ on draining channels in both models, as predicted. But it earns no more than $L^{*}$ there, and $-0.6\pm7.2\%$ and $-4.5\pm4.7\%$ on the evaluation seeds.

\section{Skew Amplitude and Forward Sizes}
\label{app:amp}

\paragraph{Skew amplitude.} The anchored skew at deployed fees with its shadow-price term scaled by $a$ (\texttt{skew\_amp}; $a=1$ is the paper's skew), on the skew tables of $K=100$ and $K=1{,}000$, seeds 0-4, paired with the deployed fees; refills on the fixed trigger, net at the table's $K$. The refill cut and gross revenue do not move with $a$.

\begin{table}[H]
\centering\scriptsize\setlength{\tabcolsep}{2.5pt}
\resizebox{\ifdim\width>\columnwidth\columnwidth\else\width\fi}{!}{%
\begin{tabular}{@{}llllll@{}}
\toprule
model & $K$ & $a$ & gross (\%) & refills (\%) & net at $K$ (k) \\
\midrule
shared & 100 & 0.5 & -0.6 & -22.3 & +3.8 $\pm$ 1.2 \\
shared & 100 & 1 & -0.9 & -21.1 & +3.6 $\pm$ 0.8 \\
shared & 100 & 2 & -1.3 & -21.7 & +3.6 $\pm$ 0.9 \\
shared & 100 & 4 & -0.8 & -23.0 & +3.9 $\pm$ 1.1 \\
shared & 100 & 8 & -0.2 & -19.5 & +3.3 $\pm$ 0.6 \\
shared & 1000 & 0.5 & -0.9 & -21.2 & +36.3 $\pm$ 6.1 \\
shared & 1000 & 1 & -0.5 & -20.5 & +35.2 $\pm$ 6.1 \\
shared & 1000 & 2 & +1.4 & -20.6 & +35.5 $\pm$ 4.7 \\
shared & 1000 & 4 & -0.9 & -20.7 & +35.3 $\pm$ 4.8 \\
shared & 1000 & 8 & +1.3 & -20.1 & +34.5 $\pm$ 4.1 \\
mc & 100 & 0.5 & +0.1 & -6.0 & +0.2 $\pm$ 0.0 \\
mc & 100 & 1 & +0.2 & -6.0 & +0.2 $\pm$ 0.0 \\
mc & 100 & 2 & +0.2 & -7.4 & +0.2 $\pm$ 0.1 \\
mc & 100 & 4 & +0.2 & -8.1 & +0.3 $\pm$ 0.1 \\
mc & 100 & 8 & +0.2 & -6.0 & +0.2 $\pm$ 0.1 \\
mc & 1000 & 0.5 & +0.1 & -7.4 & +1.6 $\pm$ 0.7 \\
mc & 1000 & 1 & +0.1 & -6.0 & +1.5 $\pm$ 1.0 \\
mc & 1000 & 2 & +0.1 & -6.6 & +1.6 $\pm$ 1.4 \\
mc & 1000 & 4 & +0.2 & -6.0 & +1.5 $\pm$ 1.1 \\
mc & 1000 & 8 & +0.1 & -8.1 & +1.9 $\pm$ 1.0 \\
\bottomrule
\end{tabular}
}
\end{table}

\paragraph{Forward sizes.} On 2026 the size-biased mean forward $m_2/m_1$ is $6\%$ of capacity at the volume-weighted median hub channel ($3$--$13\%$ interquartile; $18\%$ at the median hub channel that carries traffic).

\section{Level Search Details}
\label{app:search}

\paragraph{Forward-counting level search.} Each iteration probes $L/\sqrt2$ and $L\sqrt2$ for an epoch of payments each, counts the key's forwards (not its revenue), estimates $\hat\kappa=\ln(n_{lo}/n_{hi})/(L\sqrt2-L/\sqrt2)$, moves to $1/\hat\kappa$ (at most a factor $2$ per iteration, inside $[L_0/16,16L_0]$; no fall in forwards moves up by $2$), and holds the level once a move is under $10\%$. "joint": one level for the $40$ hubs, forwards pooled; "hub": each hub on its own forwards. $24{,}000$ payments, seeds 0-4; revenue against uniform $4{,}900$ ppm (guard26, same seeds).

\begin{table}[H]
\centering\scriptsize\setlength{\tabcolsep}{2.5pt}
\resizebox{\ifdim\width>\columnwidth\columnwidth\else\width\fi}{!}{%
\begin{tabular}{@{}llll@{}}
\toprule
model & arm & end (ppm) & rev.\ vs $L^{*}$ (\%) \\
\midrule
shared & elast\_hub\_e1500\_L1000 & 10548 & -32.6 $\pm$ 6.9 \\
shared & elast\_joint\_e1500\_L1000 & 4282 & -15.3 $\pm$ 6.8 \\
shared & elast\_joint\_e1500\_L350 & 4702 & -29.0 $\pm$ 8.7 \\
shared & elast\_joint\_e1500\_L8000 & 5630 & -5.9 $\pm$ 7.4 \\
shared & elast\_joint\_e3000\_L1000 & 5091 & -27.6 $\pm$ 3.4 \\
shared & levelsearch\_joint\_L1000 & 3380 & -18.9 $\pm$ 15.9 \\
mc & elast\_hub\_e1500\_L1000 & 8747 & -11.6 $\pm$ 10.4 \\
mc & elast\_joint\_e1500\_L1000 & 3607 & -3.5 $\pm$ 8.3 \\
mc & elast\_joint\_e1500\_L350 & 3320 & -17.0 $\pm$ 8.6 \\
mc & elast\_joint\_e1500\_L8000 & 3977 & -5.3 $\pm$ 12.1 \\
mc & elast\_joint\_e3000\_L1000 & 3488 & -11.4 $\pm$ 10.4 \\
mc & levelsearch\_joint\_L1000 & 2297 & -7.5 $\pm$ 10.5 \\
\bottomrule
\end{tabular}
}
\end{table}

\paragraph{Per-hub forward counting with long probes.} Each hub searches on its own forwards with probes of $6{,}000$ payments, four times those above, over $60{,}000$ payments from $1{,}000$ or $8{,}000$\,ppm, seeds 0-4. On average only $9$--$15$ of the $40$ hubs settle, and the hubs earn $12$--$35\%$ less than posting $L^{*}$ throughout: a hub's probe holds too few forwards for the count ratio, so the estimated elasticity is biased low and the level drifts up.

\paragraph{Noise of revenue against forward counts.} On evaluation seeds 0-9 at uniform $1{,}000$ and $4{,}900$\,ppm, the squared coefficient of variation of a hub's revenue across seeds is a median $3.2$--$5.6$ times that of its forward count (hubs with at least $5$ forwards per seed), and $4.9$--$18.5$ times for the $40$ hubs pooled: comparing forward counts needs proportionally fewer forwards than comparing revenue.

\paragraph{Level search, details.} When all $40$ hubs search at once ($24{,}000$ payments,
$5$ seeds), the mean level moves from $350$ to about $700$, from $1{,}000$ to
$1{,}300$--$1{,}400$ and from $L^{*}$ to $3{,}500$--$3{,}900$\,ppm, still
falling: where the search stops depends on where it starts. With the guard, from $L^{*}$ and $8{,}000$\,ppm it earns $5$--$10\%$ more than the unguarded search (significant with shared memory), and from every start it stays within $8\%$ of holding the start level. 

\begin{itemize}
\item Each hub posts one level $L$, probes $Lu$ and $L/u$ for $600$ payments each and moves $L$ by $u$ toward the probe that earned more on its own forwards; $24{,}000$ payments, $5$ seeds. The step $u$ starts at $1.5$ and is multiplied by $0.7$ after each reversal, down to $1.05$. Probes should be sized by own forwards until two probes differ by twice their noise.
\item Guard: a move is undone when the next pair loses more than twice the standard deviation of the hub's last eight pairs; probes stay within $[L_0/2,2L_0]$.
\item A fixed $10\%$ threshold, inside a pair's noise, undoes $150$--$240$ moves per run, not $50$, and freezes the search.
\end{itemize}

\section{Success-Constrained Tool Picks}
\label{app:succ}

\paragraph{Success-constrained tool picks.} On the calibration seeds each tool family's best ppm configuration whose mean success reaches that of uniform $1{,}000$\,ppm (shared $0.709$, mission control $0.662$), scored by gross less $100$ sat per refill; evaluated on seeds 0-9 against the Table 3 frontier.

\begin{table}[H]
\centering\scriptsize\setlength{\tabcolsep}{2.5pt}
\resizebox{\ifdim\width>\columnwidth\columnwidth\else\width\fi}{!}{%
\begin{tabular}{@{}llll@{}}
\toprule
model & pick & success & gap (\%) \\
\midrule
shared & clndabs\_150\_2400 & 0.713 & -8.3 $\pm$ 5.1 \\
shared & feeadjabs\_soft\_1200 & 0.714 & -8.1 $\pm$ 4.7 \\
shared & lndgabs\_1200 & 0.707 & -10.0 $\pm$ 2.2 \\
mc & clndabs\_300\_1200 & 0.676 & +2.9 $\pm$ 2.2 \\
mc & feeadjabs\_soft\_600 & 0.687 & +2.7 $\pm$ 1.6 \\
mc & lndgabs\_600 & 0.701 & -4.3 $\pm$ 1.3 \\
\bottomrule
\end{tabular}
}
\end{table}

\section{Welfare and Who Pays}
\label{app:welfare}

\paragraph{Welfare.} An accepted payment gives its sender an expected surplus of $v/\kappa$ and fees are transfers; a refused sender, who may still pay later, counts as a total loss, so the losses below are upper bounds. On 2026, from deployed fees, $1{,}000$\,ppm lowers total surplus,
net of the hubs' refills, by $2.0\pm0.6\%$ (shared memory) and $0.5\pm0.3\%$
(mission control), and $L^{*}$ by $7.3\pm1.3\%$ and $6.7\pm0.9\%$, while
sender surplus falls by $16$--$18\%$. The surplus lost is $0.42$ and $0.11$ of
the hubs' net gain at $1{,}000$\,ppm, and $0.78$ and $0.99$ at $L^{*}$: at the
peak the network loses about what the hubs gain.

\paragraph{Who pays.} Senders pay for the level: posting $L^{*}$ together lowers success, and payments with a detour route around the hubs (at $L^{*}$ the hubs forward $80$--$93\%$ less, the network only $16$--$18\%$).

\begin{itemize}
\item We scale revenue to the $6.6$ million routed payments of August 2023 (River): the $40$ hubs ($80\%$ of public capacity) earn $0.03$--$0.17\%$ a year on their side of their channels at deployed fees and $0.5$--$0.6\%$ at $L^{*}$: the multiples start from a low base. This is a scale check, not a calibration: it counts idle channels, at 2023 volumes.
\end{itemize}

\section{Proposition~\ref{prop:skew} Against the Exact Jump Model}
\label{app:t1}

\begin{table}[H]
\centering
\caption{Proposition~\ref{prop:skew} against the exact jump model ($K=0.05$).
ratio: $\Lambda m_1/\mu m_1^{+}$; $a$: exact slope at $b^{*}$; id.\,err: error of \eqref{eq:slope}; other columns: loss in \% of optimal revenue.}
\label{tab:t1}
\Description{Table of revenue losses against the exact jump model for the diffusion fee curve, the best constant fee with and without the refusal penalty, and the markup, on six instances.}
\scriptsize
\setlength{\tabcolsep}{2.2pt}
\begin{tabular}{@{}rrrrrrrr@{}}
\toprule
$\kappa$ & ratio & $a$ exact & id.\,err\% & ODE & best const & ($c_{\mathrm r}{=}0$) & $1/\kappa$ \\
\midrule
$5$ & $3$ & $0.0655$ & $-0.07$ & $0.11$ & $1.27$ & $1.05$ & $6.0$ \\
$5$ & $5$ & $0.0714$ & $-0.07$ & $0.01$ & $0.69$ & $0.67$ & $13.6$ \\
$10$ & $3$ & $0.0331$ & $-0.07$ & $0.18$ & $2.02$ & $1.05$ & $13.5$ \\
$10$ & $5$ & $0.0611$ & $-0.11$ & $0.04$ & $3.31$ & $2.40$ & $41.7$ \\
$20$ & $3$ & $0.0168$ & $-0.08$ & $0.29$ & $2.61$ & $1.05$ & $28.3$ \\
$20$ & $5$ & $0.0310$ & $-0.09$ & $0.06$ & $5.40$ & $2.40$ & $98.8$ \\
\bottomrule
\end{tabular}

\end{table}

\paragraph{Slope identity.} With $g$ from the jump model, which uses no evaluation data, the identity~\eqref{eq:slope} is a consistency relation, not a prediction; where the flow-balance point lies well above the refill trigger ($16$ ergodic and $8$ binding grid instances), the identity holds within $0.13\%$, and where refills supply the inventory it fails, even in sign (Table~\ref{tab:t1}).

\paragraph{A hub alone.} A hub repricing alone, all other nodes at deployed fees ($10$ seeds; the largest hub and, drawn at random, one per quartile of deployed revenue, two under mission control), gains less: under mission control $1{,}000$\,ppm raises the revenue of seven of nine hubs, by $1.2$--$2.9\times$ (one more from a $29$-sat base), and lowers two; $L^{*}$ raises four significantly and lowers three. With shared memory $1{,}000$\,ppm raises each of the five hubs tried, by $2.5\times$ or more.

\paragraph{Welfare.} An accepted payment gives its sender an expected surplus of $v/\kappa$ and fees are transfers; a refused sender, who may still pay later, counts as a total loss, so the losses below are upper bounds. On 2026, from deployed fees, $1{,}000$\,ppm lowers total surplus,
net of the hubs' refills, by $2.0\pm0.6\%$ (shared memory) and $0.5\pm0.3\%$
(mission control), and $L^{*}$ by $7.3\pm1.3\%$ and $6.7\pm0.9\%$, while
sender surplus falls by $16$--$18\%$. The surplus lost is $0.42$ and $0.11$ of
the hubs' net gain at $1{,}000$\,ppm, and $0.78$ and $0.99$ at $L^{*}$: at the
peak the network loses about what the hubs gain.

\section{Simulation Parameters}
\label{app:params}

Script and file names refer to the code release.

\paragraph{Single-channel model.} 

\begin{itemize}
\item Jump model solved exactly on a $400$-point balance grid: $c=1$, exponential forward and replenishment sizes with $m_1=m_1^{+}=0.02$, $\mu=20$, $c_{\mathrm r}=0.01$, $k=0.05$.
\item The $18$ verification instances: $\kappa\in\{5,10,20\}$, $\Lambda m_1/\mu m_1^{+}\in\{3,5\}$, $K\in\{0,0.05,0.2\}$; the six-instance table below uses $K=0.05$.
\item The $144$-instance grid: $\kappa\in\{5,10,20\}$, demand-to-replenishment ratios $1.5$--$8$, forwards of $2$--$20\%$ of capacity, cheap and costly refills, with and without $c_{\mathrm r}$.
\item Level-error ratios on that grid: at $0.8$ or $1.25$ times the best level a flat fee loses at least $2.3$ times (median $5.3$) what any skew can add, at $2/3$ or $1.5$ times at least $4.7$ times.
\item The ODE curve with $g$ measured under the flat fee $1/\kappa$ instead of the exact $g$: it loses $3.9$--$38\%$, and still $0.4$--$6.1\%$ after five more updates of $g$.
\end{itemize}

\paragraph{Traffic and calibration.} 

\begin{itemize}
\item Payment amounts: log-normal, $e^{\mathcal N(9.5,1.6^2)}$ sat, capped at $466{,}359$ sat; sources uniform, destinations drawn with probability proportional to degree$^{1.5}$; $12{,}000$ payments per run on 2026, $9{,}000$ on 2019.
\item Every policy keeps each channel's deployed base fee ($0$ or $1$\,sat on $99\%$ of hub channels).
\item At $K=100$ the deployed fees' refills cost $17.2$k (shared memory) and $2.4$k per run, at $L^{*}$ at most $0.3$k.
\item Capacities: each channel's maximum HTLC size, scaled to $30\%$ (main setting), $15\%$, $50\%$ or $100\%$.
\item Seeds. 2026: tools tuned on calibration seeds $1000$--$1002$, evaluated on $0$--$9$ (some settings $0$--$19$); 2019: $1000$--$1004$ and $0$--$19$. Skew tables fitted on $1000$--$1004$, refitted for the wider-coverage test on those plus $2000$--$2044$ ($600$k payments).
\item A report-only two-part fit $\lambda(L)=\Lambda[(1-w)e^{-(\kappa_c+200)L}+we^{-200L}]$ to the scan gives a captive share $w=0.44$ (2026) and $0.42$ (2019).
\item The calibration scan for $L^{*}$ posts $15$ uniform levels ($25$--$12{,}000$\,ppm) with balances frozen; $L^{*}$ is the vertex of a parabola in $\log L$ through the best grid point and its neighbours. Frontier gaps use per-seed linear interpolation between evaluated levels, never the parabola.
\item Skew tables: a $100$-point balance grid, refill trigger $5\%$ of capacity, refill target $50\%$, $K=100$ (refitted at $K=1{,}000$ for the $K=1{,}000$ runs), solved by relative value iteration; a channel gets a table with at least $20$ outgoing and $10$ incoming calibration forwards ($10$ and $5$ in the wider-coverage refit).
\end{itemize}

\paragraph{Routing models.} 

\begin{itemize}
\item Shared memory: every sender avoids a failed channel, for that amount, during the next $300$ payments.
\item Mission control: each sender records each node pair's last result. A pair's success probability for an amount is $0.95$ up to a recorded success's amount; from a recorded failure's amount up, it recovers from $0$ toward the prior with a half-life of $300$ payments; otherwise it is the node prior, $0.6$ with weight $2$, mixed with the probabilities of the node's recorded pairs. A route minimizes the sum of fees plus an attempt cost of $100$\,sat $+\,1{,}000$\,ppm times $-\ln$ of each hop's probability (LND's defaults).
\item Sensitivity: a half-life of $100$ or $1{,}000$ payments changes every policy's revenue by less than $0.1\%$ ($5$ seeds).
\end{itemize}

\paragraph{Tools.} 

\begin{itemize}
\item \texttt{charge-lnd} posts $\min+(1-x)(\max-\min)$ for local share $x$, with $\min,\max$ in ppm or as multiples of the deployed rate; checked when the share moves by $5\%$, it re-posts on changes of at least $10$\,ppm.
\item feeadjuster multiplies the base rate by $50^{0.5-x}$ (or its soft or hard variants) and re-posts on $5\%$ share moves.
\item LNDg adjusts each channel once per "day" of $300$ payments by the tier of its balance share and its one- and seven-day volume; its failed-HTLC rule is on in the ppm runs (default LNDg is unchanged), and its fee cap, $2{,}500$\,ppm by default, is $10{,}000$\,ppm, or $4L$ for starts $L\ge4{,}000$\,ppm.
\item ppm grids center on $L\in\{600,1{,}200,2{,}450,4{,}900\}$ with ranges $(L/2,2L)$, $(L/4,4L)$, $(0,2L)$ (feeadjuster soft and default); 2019: $L\in[250,12{,}000]$. Multiples on 2026: $24$ configurations up to $1{,}024\times$. Extra centers of $9{,}800$ and $12{,}000$\,ppm were added later (App. B).
\item Defaults (\S{}7.4): \texttt{charge-lnd} from $0.25\times$ to $4\times$ the deployed rate, feeadjuster and LNDg as shipped, all starting from deployed fees.
\end{itemize}

\paragraph{Realism options.} 

\begin{itemize}
\item With every realism option on, removing refills changes revenue by $-0.3\%$ and the skew by $+0.2\%$.
\end{itemize}

\section{Confirmatory Tests, Frozen Predictions and Power}
\label{app:stats}

Two claims, both on the main 2026 setting: (1) no ppm-tuned tool beats the frontier by more than 3\%; (2) the skew adds less than 3\% under mission control and 5\% with shared memory. One-sided paired t-tests over the 10 evaluation seeds, three tools and three skew comparisons per routing model, Holm correction over all 12, family-wise alpha = 0.05. All 12 null hypotheses are rejected; the largest Holm-adjusted p is 0.0086 (\texttt{chan/holm12.json}, all 12 rows with raw and adjusted p; the tool tests use the frontier over every uniform level run, including 2,450, 6,125 and 12,000 ppm, which Table 3's frontier omits; on Table 3's frontier the mission-control raw p change (4.8e-07 -> 7.5e-07, 8.5e-07 -> 2.5e-06, 6.7e-04 -> 7.4e-04), the shared-memory ones do not, and no conclusion does; the 10 tool comparisons against H0 gap $\leq$ 0 are exploratory and not in the family).

\begin{table}[H]
\centering\scriptsize\setlength{\tabcolsep}{2.5pt}
\resizebox{\ifdim\width>\columnwidth\columnwidth\else\width\fi}{!}{%
\begin{tabular}{@{}llllll@{}}
\toprule
model & comparison & H0 & mean & p & p (Holm, 12) \\
\midrule
shared & charge-lnd (ppm) & gap $\geq$ +3\% & -4.17\% & 0.0014 & 0.0065 \\
shared & feeadjuster (ppm) & gap $\geq$ +3\% & -3.96\% & 0.0086 & 0.0086 \\
shared & LNDg (ppm) & gap $\geq$ +3\% & -3.17\% & 0.0016 & 0.0065 \\
shared & skew at deployed fees & gain $\geq$ 5\% & -1.79\% & 0.0013 & 0.0065 \\
shared & skew on 1,000 ppm & gain $\geq$ 5\% & -0.74\% & 3.4e-05 & 0.00024 \\
shared & skew on $L^{*}$ & gain $\geq$ 5\% & +0.04\% & 0.002 & 0.0065 \\
mc & charge-lnd (ppm) & gap $\geq$ +3\% & -10.55\% & 4.8e-07 & 4.3e-06 \\
mc & feeadjuster (ppm) & gap $\geq$ +3\% & -7.57\% & 8.5e-07 & 6.8e-06 \\
mc & LNDg (ppm) & gap $\geq$ +3\% & -1.50\% & 0.00067 & 0.004 \\
mc & skew at deployed fees & gain $\geq$ 3\% & +0.19\% & 1.8e-09 & 1.8e-08 \\
mc & skew on 1,000 ppm & gain $\geq$ 3\% & -0.04\% & 1.2e-13 & 1.3e-12 \\
mc & skew on $L^{*}$ & gain $\geq$ 3\% & +0.01\% & 1e-20 & 1.2e-19 \\
\bottomrule
\end{tabular}
}
\end{table}

Not part of the family: no tool, including those tuned as multiples of deployed fees, rejects H0 gap $\leq$ 0 (p $\geq$ 0.92 raw), so none lies significantly above the frontier.

\paragraph{Frozen hub-count predictions.} Frozen before the runs, its predictions for $10$, $20$ and $80$ hubs held: under shared memory each peaks at $4{,}900$\,ppm, the grid level nearest the prediction, and stays within $0.07$ of the predicted curve; under mission control the predicted level earns at least $94\%$ of the best (\texttt{predictions\_hubk.json}).

\paragraph{Tools against their center.} Against the flat level at its center, paired by seed, LNDg earns $-3.1\pm5.0\%$ and $+0.2\pm2.3\%$, \texttt{charge-lnd} $-5.2\pm4.6\%$ and $-2.2\pm3.3\%$, feeadjuster $-20.4\pm4.8\%$ and $-8.8\pm1.8\%$, for at most $4.2$ points more success.

One-sided paired t-tests over the 10 evaluation seeds (\texttt{chan/power.json}, \texttt{analysis/power\_confirmatory.py}). Power is for a true effect of 0; the last column is the smallest margin rejected with 80\% power at the strictest Holm level (alpha/12).

\begin{table}[H]
\centering\scriptsize\setlength{\tabcolsep}{2.5pt}
\resizebox{\ifdim\width>\columnwidth\columnwidth\else\width\fi}{!}{%
\begin{tabular}{@{}lllllll@{}}
\toprule
test & mean (\%) & se & margin & power ($\alpha$/12) & power ($\alpha$) & margin at 80\% \\
\midrule
shared/clndabs & -4.17 & 1.77 & 3 & 0.11 & 0.47 & 7.7 \\
shared/feeadjabs & -3.96 & 2.39 & 3 & 0.05 & 0.31 & 10.4 \\
shared/lndgabs & -3.17 & 1.56 & 3 & 0.14 & 0.55 & 6.8 \\
shared/skew@deployed & -1.79 & 1.65 & 5 & 0.43 & 0.88 & 7.1 \\
shared/skew@1000ppm & -0.74 & 0.83 & 5 & 0.98 & 1.00 & 3.6 \\
shared/skew@$L^{*}$ & +0.04 & 1.29 & 5 & 0.68 & 0.97 & 5.6 \\
mc/clndabs & -10.72 & 1.17 & 3 & 0.29 & 0.76 & 5.1 \\
mc/feeadjabs & -8.04 & 1.01 & 3 & 0.41 & 0.86 & 4.4 \\
mc/lndgabs & -1.59 & 1.00 & 3 & 0.42 & 0.87 & 4.3 \\
mc/skew@deployed & +0.19 & 0.13 & 3 & 1.00 & 1.00 & 0.5 \\
mc/skew@1000ppm & -0.04 & 0.05 & 3 & 1.00 & 1.00 & 0.2 \\
mc/skew@$L^{*}$ & +0.01 & 0.01 & 3 & 1.00 & 1.00 & 0.0 \\
\bottomrule
\end{tabular}
}
\end{table}

Probe size: forwards a hub needs per level to tell two levels 4\% apart (two-sided 5\%, 80\% power), from the seed-to-seed variation of its revenue at a fixed level.

\begin{table}[H]
\centering\scriptsize\setlength{\tabcolsep}{2.5pt}
\resizebox{\ifdim\width>\columnwidth\columnwidth\else\width\fi}{!}{%
\begin{tabular}{@{}llllll@{}}
\toprule
model/level & hubs & CV & fwd/run & needed (median) & interquartile \\
\midrule
shared/uniform1000 & 39 & 0.38 & 21.8 & 32,186 & 16,597--41,067 \\
shared/uniform4900 & 37 & 0.62 & 11.0 & 35,473 & 26,056--44,234 \\
mc/uniform1000 & 40 & 0.24 & 53.8 & 32,934 & 24,660--42,042 \\
mc/uniform4900 & 39 & 0.51 & 19.1 & 39,672 & 25,034--56,721 \\
\bottomrule
\end{tabular}
}
\end{table}

\section{Ordering Across 2026 Settings}
\label{app:settings}

\begin{table}[H]
\centering
\caption{The ordering across 2026 settings, senders blind to balances, each setting
with its own calibration ($10$ seeds): revenue
peak, hub gross revenue and refills at $L^{*}$ relative to deployed fees, and
the refill change from the skew at deployed fees. Refills are forced except with economic refills and all realism options ($4.4$, $6.6$ and $5.0$ per run at deployed fees, $L^{*}$ and with the skew).
$^{\dagger}$Best level on the evaluation seeds, so the multiple is an upper bound ($\kappa{=}100$ and capacities $15$, $50\%$: $5$ seeds). --: no refills at deployed fees, or skew not run.}
\label{tab:robust}
\Description{Table of the revenue peak, revenue multiple and refill change at the peak, and the skew's refill change, for each robustness setting on the 2026 snapshot; the 2019 settings are summarized in the appendix.}
\scriptsize
\setlength{\tabcolsep}{2.4pt}\renewcommand{\arraystretch}{0.86}
\begin{tabular}{@{}lrrr>{\color{trimblue!80!black}}r@{}}
\toprule
Setting & $L^{*}$ & Rev. & Refills & \textbf{Skew} \\
\midrule
2026, shared memory & $4{,}900$ & $17.1\times$ & $-98\%$ & $-21\%$ \\
2026, mission control & $4{,}900$ & $3.4\times$ & $-91\%$ & $-5\%$ \\
2026, economic refills, shared & $4{,}900$ & $13.1\times$ & -- & -- \\
2026, economic refills, MC & $4{,}900$ & $3.4\times$ & -- & -- \\
2026, all realism options & $4{,}900$ & $3.2\times$ & $+50\%$ & $+14\%$ \\
2026, MC, payments $5\times$ & $4{,}000^{\dagger}$ & $4.3\times$ & $-92\%$ & $-7\%$ \\
2026, shared, payments $5\times$ & $4{,}900^{\dagger}$ & $20.9\times$ & $-98\%$ & $-23\%$ \\
2026, MC, capacities $15$; $50\%$ & $4{,}000^{\dagger}$ & $3.3$; $3.4\times$ & $-88$; $-94\%$ & -- \\
2026, shared, capacities $15$; $50\%$ & $4{,}900^{\dagger}$ & $17.6$; $18.8\times$ & $-99$; $-100\%$ & -- \\
2026, MC, $\kappa{=}100$ & $8{,}000^{\dagger}$ & $5.7\times$ & $-85\%$ & -- \\
2026, MC, $\kappa{=}400$ & $1{,}500^{\dagger}$ & $2.5\times$ & $-80\%$ & -- \\
2026, MC, $\kappa{=}1000$ & $600^{\dagger}$ & $2.3\times$ & $-65\%$ & -- \\
2026, MC, $\kappa{=}2000$ & $500^{\dagger}$ & $2.2\times$ & $-76\%$ & -- \\
2026, shared, $\kappa{=}2000$ & $500^{\dagger}$ & $4.3\times$ & $-99\%$ & -- \\
2026, MC, heavy-tailed WTP & $3{,}000^{\dagger}$ & $2.7\times$ & $-90\%$ & -- \\
2026, shared, heavy-tailed WTP & $5{,}000^{\dagger}$ & $13.6\times$ & $-98\%$ & -- \\
\bottomrule
\end{tabular}
\end{table}

\paragraph{Ordering across settings.} With LND's route cost, mission control ranks routes by fee plus the attempt cost divided by the route's success probability, as LND does. With wallet fee limits, senders accept a route only within their wallet's default fee limit (CLN $0.5\%$, LND $5\%$; half the senders each), and the limits place the peak. With all realism options on, refills are made only when they pay (App.~\ref{app:x26}). At full capacities ($10$ seeds) LNDg earns $2.1\pm2.9\%$ and $2.8\pm2.0\%$ more than the frontier at its success; the other tools are within the noise of it or below. With every realism option on, where refills are made only when they pay, the skew raises refills slightly ($4.4$ to $5.0$ per run at deployed fees), since higher fees make refills pay. The best tools lie $2$--$11\%$ below the frontier with $\kappa=2{,}000$, where a uniform level still earns $2.2$ and $4.3\times$ the deployed fees, $5$--$6\%$ with a heavier-tailed willingness to pay, $6$--$7\%$ with payments five times larger, and $6$--$8\%$ at capacity proxies of $15\%$, within $3\%$ at $50\%$. At their defaults the tools keep the deployed level, earning $-42$ to $+8\%$ of its revenue, $40$--$49\%$ below the frontier under mission control.

\paragraph{Wallet-count sweep.} Mission control kept per wallet, senders hashed to wallets; seeds 0-4 (20 wallets: 0-9).

\begin{table}[H]
\centering\scriptsize\setlength{\tabcolsep}{2.5pt}
\resizebox{\ifdim\width>\columnwidth\columnwidth\else\width\fi}{!}{%
\begin{tabular}{@{}lll@{}}
\toprule
wallets & best level & multiple \\
\midrule
1 & 4000 ppm & 3.54 \\
5 & 4000 ppm & 3.34 \\
20 & 4000 ppm & 3.18 \\
100 & 4000 ppm & 3.15 \\
\bottomrule
\end{tabular}
}
\end{table}

\end{document}